\documentclass[aps,prx,twocolumn,superscriptaddress,numerical]{revtex4-2}
\usepackage[utf8]{inputenc}
\usepackage[english]{babel}
\usepackage[T1]{fontenc}
\usepackage{amsthm,bm}
\usepackage{physics}
\usepackage{subfigure}
\usepackage{amsfonts} 
\usepackage{mathtools}
\usepackage{dsfont}
\usepackage{hyperref}
\usepackage{graphicx}   
\usepackage{amsmath, amssymb}  
\usepackage{soul} 
\usepackage{float}      
\usepackage{caption}    
\usepackage{subcaption} 
\makeatletter
\def\@makecaption#1#2{%
  \par
  \begingroup
  \small
  \setlength{\parindent}{0pt}
  \rightskip=0pt
  \leftskip=0pt
  \parfillskip=0pt plus 1fil
  \spaceskip=0pt
  \xspaceskip=0pt
  \noindent\textbf{#1.} #2\par
  \endgroup
}
\makeatother

\usepackage{booktabs}

\usepackage{tikz}
\usepackage{lipsum}

\newtheorem{theorem}{Theorem}

\newtheorem{corollary}{Corollary}

\newcommand{\utheta}{\underline{\theta}}

\begin{document}

\title{Resource-Optimal Importance Sampling for  Randomized Quantum Algorithms}

\author{Davide Cugini}
\affiliation{Dipartimento di Fisica ``Alessandro Volta,'' Università di Pavia, via Bassi 6, 27100 Pavia, Italy}
\affiliation{Theoretical Division, Los Alamos National Laboratory, Los Alamos, NM 87545, USA}

\author{Touheed Anwar Atif}
\affiliation{Computing and Artificial Intelligence Division, Los Alamos National Laboratory, Los Alamos, NM 87545, USA}

\author{Yi\u{g}it Suba\c{s}\i}
\affiliation{Computing and Artificial Intelligence Division, Los Alamos National Laboratory, Los Alamos, NM 87545, USA}

\begin{abstract}
Randomized protocols are procedures that incorporate probabilistic choices during their execution
and they play a central role in quantum algorithms, 
spanning Hamiltonian simulation, noise mitigation, and measurement tasks. In practical implementations, the dominant cost of such protocols typically arises from circuit execution and measurement, and depends on hardware-specific resources such as gate counts, circuit depth, runtime, or dissipated energy.
We introduce a general framework for applying classical importance sampling to randomized quantum protocols. Given a cost function for running quantum circuits, the proposed approach minimizes a net-cost figure of merit that jointly captures the computational expense per circuit and the estimator variance. 
We further extend the framework to scenarios where the quantum computation is subject to errors arising either from algorithmic approximations or from physical noise,  
proving that importance sampling preserves estimator bias despite altering the sampling distribution, and to settings with error-detection schemes, where we characterize the resulting changes in the optimal sampling strategy and achievable net-cost reduction. Representative applications include the \textsc{Qdrift} protocol, dephasing channels, mixed-states simulation, composite observables estimation, classical shadows, and probabilistic error cancellation. Overall, our results establish a principled approach for reducing the computational resources required by randomized quantum protocols through classical sampling optimization.
\end{abstract}

\maketitle

\newcommand{\DC}[1]{{\color{violet}[Davide: #1]}}
\newcommand{\TAcom}[1]{{\color{blue}[Touheed: #1]}}
\newcommand{\TA}[1]{{\color{green}#1}}
\newcommand{\ys}[1]{{\color{blue}#1}}
\newcommand{\cut}[1]{{\color{red}[Delete: #1]}}

\section{Introduction}
Randomized methods have become integral to quantum algorithms and protocols.
Broadly speaking, 
randomized algorithms are procedures that incorporate probabilistic choices during their execution~\cite{mitzenmacher2017probability}. 
In the context of digital quantum computation, 
this typically entails sampling a subset of quantum operations in each run of a circuit
according to a classical probability distribution.

This paradigm arises in several areas of quantum information processing.  
In the context of Hamiltonian simulation, randomized product formulas can provide tighter error bounds than fixed-order Trotter-Suzuki schemes, whereas the \textsc{Qdrift} protocol follows an alternative strategy, approximating the time-evolution operator $e^{-iHt}$ by sampling Hamiltonian terms proportionally to their operator norms~\citep{Childs2019,Campbell2019}.
In the context of error reduction, 
~\cite{Wallman2016} introduces randomized compiling, which inserts suitable random gates to tailor coherent errors into an effectively stochastic (Pauli) error model compatible with fault-tolerant error correction.
In quantum characterization and measurement, 
protocols such as direct fidelity estimation (DFE) and classical shadows 
rely explicitly on randomized measurements, 
with DFE already incorporating importance weighting over Pauli operators 
to allocate sampling effort where it is most informative~\citep{FlammiaLiu2011,HuangKuengPreskill2020}.
Beyond protocols that are intrinsically randomized, 
many deterministic quantum algorithms admit randomized reformulations; 
this perspective is relevant, for instance, 
in the simulation of mixed states and in the estimation of composite observables, 
where stochastic decompositions can be leveraged to construct efficient sampling-based procedures~\cite{arrasmith2020operator}.

Given the stochastic nature of randomized protocols,
a natural question arises: can one exploit the classical technique of Importance Sampling (IS) in the quantum setting to improve the performance of randomized quantum algorithms? 
IS consists of replacing the original probability distribution $p$ of a randomized protocol with an alternative distribution $q$, while compensating for this modification by reweighting each measurement outcome with an appropriate factor $w$,
so that the estimator bias remains unchanged.
In classical computation, 
IS is typically employed for two main purposes,
namely
i) variance reduction:
if approximate prior knowledge of the desired result is available, 
a suitable choice of the distribution $q$ can significantly decrease the variance of the estimator 
and ii) cost reduction: when drawing samples from $p$ constitutes the dominant computational burden, IS allows one to replace it with a more convenient distribution $q$. In this case,
the resulting estimator remains unbiased, but its variance generally changes. 
Therefore, IS amounts to balancing the sampling cost against the variance of the estimator~\citep{OwenMCBook,LiuMCBook}.

In this work, we adopt a problem-agnostic perspective, 
assuming no prior knowledge of the computation’s outcome, 
and therefore focus on the second classical use of IS: 
reducing the average computational resources.
For quantum randomized protocols, 
the primary contribution to the overall cost typically does not arise from classical sampling. 
Instead, it comes from the “quantum sampling” performed via measurements at the end of the circuit, 
which inherently includes the cost of implementing the circuit itself.
Depending on the theoretical framework and hardware constraints, the relevant cost metric may correspond to circuit depth, gate counts (e.g., CNOT or T gates), or even physical resources such as execution time or dissipated energy.
The first step of a randomized quantum algorithm 
is the sampling from a classical distribution, the outcome of which determines which quantum circuit is ran on the device. If the cost of executing the circuits is different, 
IS can be used to modify the sampling distribution so that more expensive circuits are sampled less frequently, 
while cheaper circuits are sampled more often.
However, 
as in the classical setting, this changes the variance of the estimator, which in the quantum case is determined by the statistical fluctuations of single-shot measurement outcomes obtained at the end of each circuit execution.
In general, the variance of the estimator may increase, necessitating more independent runs to reach a desired precision. That is, while IS can reduce the expected cost per run, it may increase the total number of runs required to achieve a given accuracy. In this work, we show how to select the sampling distribution $q$ to optimally balance this tradeoff. \\

The manuscript is structured as follows.
In Section~\ref{sec: General Discussion} we analytically derive the optimal IS distribution $q^*$ that minimizes the net cost
\begin{align}
NC \;=\;& \bigl(\text{average per-run resources}\bigr)\times \nonumber \\
&\;\; \bigl(\text{variance of the estimator}\bigr)\,, \nonumber
\end{align}
for a given original sampling distribution $p$ characterizing the randomized protocol and a cost function associated with quantum circuit evaluations.
We then extend the analysis to scenarios where the quantum computation is subject to imperfections
arising either from algorithmic approximations or from physical noise. 
Within this setting, the reduction in expected per-run cost enabled by importance sampling could suggest a potential mitigation of such effects. 
However, in Section~\ref{sec:IS_imperfect_channels} we demonstrate that the estimator bias introduced by algorithmic or physical imperfections is unaltered by IS. 
In Section~\ref{sec : IS with error detection}, 
we further extend the previous results to scenarios in which the quantum circuit is equipped with an error-detection scheme,
showing how the optimal sampling distribution $q^*$ 
and the achievable reduction in net cost are modified in this setting.
In Section~\ref{sec:applications} we report representative applications of our framework. 
In particular, 
we discuss its implications for the \textsc{Qdrift} protocol~\citep{Campbell2019}, 
the implementation of dephasing channels~\citep{BoixoKnillSomma2009,CunninghamRoland2024}, 
the mixed-state simulation and 
the composite observable estimation, 
the classical shadow tomography~\citep{HuangKuengPreskill2020}, 
and the Probabilistic Error Cancellation (PEC) technique for error mitigation~\citep{TemmeBravyiGambetta2017,EndoBenjaminLi2018,TakagiLimitsQEM2022,RMPreviewQEM2023}.
The results obtained demonstrate how optimal IS enables a reduction in computational resources solely through the manipulation of classical sampling distributions, 
without requiring any modification of the underlying quantum protocols.

\section{Optimal IS}\label{sec: General Discussion}
In this work, we address the problem of efficiently estimating the expectation value $\Tr[O\, \rho_\text{target}]$
of an observable $O$ on a target quantum state $\rho_\text{target}$. 
We focus on scenarios in which $\rho_\text{target}$ admits a representation as an ensemble average over a family 
of more easily implementable quantum states, obtained by applying quantum channels 
$\mathcal{E}_{\underline{\theta}}$ to a fixed reference state $\rho$, with sampling probabilities given by a probability distribution $p(\underline{\theta})$. 
Namely,
\begin{align}\label{eq: state decomposition}
    \rho_\text{target} 
    = \int d\underline{\theta}\; p(\underline{\theta})\, \mathcal{E}_{\underline{\theta}}(\rho)
    = \mathbb{E}_{p}\!\left[\,\mathcal{E}_{\underline{\theta}}(\rho)\,\right].
\end{align}
Consider an observable
\begin{align}
    O = \sum_x x \,\Pi_x\,,
\end{align} 
where $\Pi_x$ is a projector on the eigenspace associated with its eigenvalue $x$.
One can estimate the expectation value of $O$ 
on the state $\rho_\text{target}$ by preparing $\mathcal{E}_{\underline{\theta}}(\rho)$ 
with probability $p(\underline{\theta})$ and measuring in the eigenbasis of $O$. 
Such an estimator is unbiased, indeed
\begin{align}
    \Tr[O \rho_\text{target}] &= \int d\underline{\theta} \,p(\underline{\theta}) \sum_x \Tr[\Pi_x\, \mathcal{E}_{\underline{\theta}}(\rho)] x \nonumber\\
    &= \sum_x \int d\underline{\theta} \,\left(p(\underline{\theta})  r(x|\underline{\theta})\right)\, x \nonumber \\
     &= \mathbb{E}_{pr}\left[x\right]  \,,
\end{align}
where $r(x|\underline{\theta}):= \Tr[\Pi_x\,\mathcal{E}_{\underline{\theta}}(\rho)]$ 
is the probability of obtaining the measurement outcome $x$
given the application of the channel $\mathcal{E}_{\utheta}$, 
while $\mathbb{E}_{pr}$ denotes the expectation over the joint distribution $p(\underline\theta)r(x|\underline{\theta})$.
Here and in the following, we assume that the operator norm of the observable satisfies $\norm{O} \le 1$; 
if this is not the case, one can always enforce it by rescaling $O$ by its operator norm.
The total computational resources needed for the estimation is determined by the number of independent samples required to approximate the expectation value of $O$ to a desired accuracy $\delta$, as well as the resources needed to prepare and measure each instance of $\mathcal{E}_{\underline{\theta}}(\rho)$. We quantify the latter by $c(\underline{\theta})$. Note that the single-sample cost $c(\underline{\theta})$ may be defined in terms of any relevant computational resource, such as runtime, circuit depth or the total number of two-qubit gates in a quantum circuit.
While the variance of the estimator imposes a lower bound 
on the number of samples needed, 
the expected cost per sample depends on 
the sampling distribution $p(\underline{\theta})$. 
We therefore explore the use of IS
as a means to optimize this trade-off. 
By sampling from another distribution 
$q(\underline{\theta})$ 
and multiplying the measurement outcome of $\mathcal{E}_{\underline{\theta}}(\rho)$ by a corresponding weight $w(\underline{\theta}) := p(\underline{\theta})/q(\underline{\theta})$,
one can in principle reduce the average computational resources per run,
while keeping the expectation value of the estimation unchanged.
Indeed
\begin{align}
    \mathbb{E}_{qr}\left[w(\utheta)x\right] 
= \mathbb{E}_{pr}\left[x\right] .
\end{align}
However,
although the variance for a specific observable may decrease or remain comparable, 
the variance of the worst-case observable (i.e., the upper bound over all possible observables) 
is generally increased by IS. 
To make these considerations quantitative, let us define the expected cost per run under the modified sampling as
\begin{align}
\mathbb{E}_{qr}[c(\underline{\theta})] &= \sum_x \int d\underline{\theta}q(\underline{\theta})r(x|\underline{\theta})\, c(\underline{\theta}) \nonumber \\  
&= \int d\underline{\theta}\, q(\underline{\theta})\, c(\underline{\theta}) \nonumber\\
&= \mathbb{E}_{q}[c(\underline{\theta})] \,,
\end{align}
where $c(\underline{\theta})$ is the cost
associated with the implementation of the channel $\mathcal{E}_{\underline{\theta}}(\rho)$ and subsequent measurement of $O$.
Then, the maximum variance over all possible observables is bounded by
\begin{align}
    \mathrm{Var}^{\max}_{qr} 
    &= \max_{|x| \leq 1} \left\{\mathbb{E}_{qr}\left[(w(\underline{\theta}) x)^2\right]\nonumber - \left(\mathbb{E}_{qr}\left[w(\underline{\theta}) x\right]\right)^2\nonumber \right\}\\
    &\leq \max_{|x| \leq 1}\, \mathbb{E}_{qr}\left[(w(\underline{\theta}) x)^2\right]\nonumber \\
    &= \mathbb{E}_q\left[w^2(\underline{\theta})\right]\,, 
\end{align}
Hence, by means of the Central Limit Theorem, 
the total number of independent samples required to achieve a desired precision $\delta$ 
scales as $\mathbb{E}_q[w^2(\underline{\theta})]\,\delta^{-2}$.
For this reason, we introduce a new quantity,
which we call net-cost of the computation,
defined as
\begin{equation}\label{eq: net cost definition}
    NC_q = \mathbb{E}_q[c(\underline{\theta})] \, \mathbb{E}_q[w^2(\underline{\theta})]\,,
\end{equation}
With this definition the expected cost for desired precision $\delta$ scales as $\delta^{-2} NC_q $, independently of the observable. When we don't use IS, the net cost is simply
\begin{equation}\label{eq: net cost without IS}
    NC_p = \mathbb{E}_p[c(\underline{\theta})] \,,
\end{equation}
since in this case $w(\theta)=1$ and hence so is the worst-case variance.

Although using an arbitrary IS distribution $q$ may increase the worst-case variance, 
a suitable choice of $q$ can, in principle, reduce the overall net cost $NC_q$ compared to the conventional sampling scheme. 
In particular, we are interested in identifying the optimal distribution $q$, 
i.e., the one that minimizes $NC_q$. 
The following theorem provides a characterization of this optimal choice.

\begin{theorem}[Optimal Sampling Distribution]\label{thm:optSampling}
Let $p(\underline{\theta})$ be a probability distribution, $c(\underline{\theta})$ a cost function, and $q(\underline{\theta})$ an arbitrary IS distribution. Then,
\begin{equation}
    NC_q \ge \left( \mathbb{E}_p\!\left[c^{1/2}(\underline{\theta})\right] \right)^2, \quad \forall q,
\end{equation}
with equality if and only if
\begin{equation}
    q(\underline{\theta}) \propto \frac{p(\underline{\theta})}{c^{1/2}(\underline{\theta})}.
\end{equation}
\end{theorem}

\begin{proof}
We begin with the definition of the net cost~\eqref{eq: net cost definition} and note that its dependence on $q$ appears both explicitly in the sampling distribution and implicitly through the weight $w$.  
Using the change-of-measure identities 
$\mathbb{E}_q[c(\underline{\theta})] = \mathbb{E}_p[c(\underline{\theta})\,w^{-1}(\underline{\theta})]$ and 
$\mathbb{E}_q[w^2(\underline{\theta})] = \mathbb{E}_p[w(\underline{\theta})]$, 
we obtain
\begin{align}
    NC_q = \mathbb{E}_p[c(\underline{\theta})\,w^{-1}(\underline{\theta})]\,\mathbb{E}_p[w(\underline{\theta})]\,.
\end{align}
Applying the Cauchy--Schwarz inequality yields
\begin{align}
    NC_q \ge \left( \mathbb{E}_p\!\left[c^{1/2}(\underline{\theta})\right] \right)^2,
\end{align}
with equality if and only if both $\mathbb{E}_p[c(\underline{\theta})\,w^{-1}(\underline{\theta})]$ and $\mathbb{E}_p[w(\underline{\theta})]$ are proportional to $\mathbb{E}_p[c^{1/2}(\underline{\theta})]$, 
which occurs when $w(\underline{\theta}) \propto c^{1/2}(\underline{\theta})$. 
The result directly follows.
\end{proof}

Besides providing the optimal distribution $q$, which we henceforth denote by $q^*$, 
the previous result can also be used to quantify the advantage of employing optimal IS over naively sampling from $p$. 
This advantage can be expressed as
\begin{align}\label{eq:netCostRatio}
    \frac{NC_{q^*}}{NC_p}
    &= \frac{\left(\mathbb{E}_p[c^{1/2}(\underline{\theta})]\right)^2}{\mathbb{E}_p[c(\underline{\theta})] } \nonumber\\
    &= \left[1 + \frac{\mathrm{Var}_p\!\left[c^{1/2}(\underline{\theta})\right]}{\left(\mathbb{E}_p[c^{1/2}(\underline{\theta})]\right)^2}\,\right]^{-1}.
\end{align}
This expression clearly shows that $NC_{q^*} \leq NC_p$. 
Moreover, by treating $c^{1/2}(\underline{\theta})$ as a random variable with $\underline{\theta} \sim p$, We see that the advantage grows with the ratio between the variance and the squared mean of $c^{1/2}$..
This provides an intuitive interpretation of the effectiveness of optimal IS: its benefit is more pronounced when the distribution of the square root of the cost function is broad rather than sharply peaked.
We conclude this section with two remarks. 
First, although the magnitude of the net-cost reduction depends on the specific problem instance,
switching from the original sampling distribution $p$ to the optimal IS distribution $q^*$ 
does not modify the structure of the quantum circuit. 
The same family of channels $\mathcal{E}_{\underline{\theta}}$ must still be implemented, 
and the measurement procedure remains unchanged; only the classical distribution used to sample the parameters $\underline{\theta}$ is altered.
Second, all the results and derivations presented above apply equally well when the parameters $\underline{\theta}$ take discrete values, without any modification to the formalism.

\subsection{Composite Quantum Channels}
\label{sec:multilayer}
In the following, we focus on a setup that satisfies:
\begin{enumerate}
    \item The random vector of variables $\underline{\theta}$ has dimension $s$ 
    and its components $\theta_i$ are drawn independently. 
    \item The cost is additive:
    \begin{align}
        c(\underline{\theta}) = \sum_{i=1}^s c_i(\theta_i),
    \end{align}
    i.e., the total cost is the sum of individual contributions $c_i$ associated with each random variable.
\end{enumerate}
Notice that each variable $\theta_i$ may represent either a single scalar parameter 
or a whole collection of parameters.
This scenario frequently arises, for example when the channel $\mathcal{E}_{\underline{\theta}}$ is given by a sequential composition
\begin{align}
    \mathcal{E}_{\underline{\theta}} 
    = \mathcal{E}_{\theta_s} \circ \mathcal{E}_{\theta_{s-1}} \circ \cdots \circ \mathcal{E}_{\theta_1},
\end{align}
of quantum channels $\mathcal{E}_{\theta_i}$, 
each depending on its own independent parameter $\theta_i$.
For instance, 
this is the behavior of the Eigenpath Traversal \cite{BoixoKnillSomma2009} 
and \textsc{Qdrift}-type protocols \cite{campbell2018random} on an ideal quantum device.
Let $\{\theta_i\}$ be independent random variables. 
Then the quantities $c_i := c_i(\theta_i)$ can themselves be regarded as independent random variables
with mean $\mu_i$ and variance $\sigma_i^2$.
The total cost $c(\underline{\theta})$
is then a sum of independent random variables whose distribution approaches a Gaussian for large $s$ (central limit theorem), with mean and variance
\begin{align}
    \mu &= \sum_{i=1}^s \mu_i := s \bar{\mu},\\
    \sigma^2 &= \sum_{i=1}^s \sigma_i^2 := s \bar{\sigma}^2,
\end{align}
where $\bar{\mu}$ and $\bar{\sigma}^2$ denote the average mean and variance, 
respectively.  
Assuming that both $\bar{\mu}$ and $\bar{\sigma}^2$ converge to a finite value 
for $s \to \infty$,
we find (see Appendix~\ref{thm:TheoremIIDLInfty}) that
\begin{align}
    \mathbb{E}_{p}[c(\underline{\theta})] 
        &\sim s\,\bar{\mu} + \mathcal{O}(\sqrt{s}),\\
    \big(\mathbb{E}_{p}[c^{1/2}(\underline{\theta})]\big)^2
        &\sim s\,\bar{\mu} - \frac{\bar{\sigma}^2}{4\bar{\mu}} + \mathcal{O}(s).
\end{align}
Thus, the ratio in Eq.~\eqref{eq:netCostRatio} becomes
\begin{equation}
    \frac{NC_{q^*}}{NC_p}
    = 1 - \frac{1}{s}
    \left(\frac{\bar{\sigma}}{2\bar{\mu}}\right)^2
    + o\!\left(\frac{1}{s}\right),
\end{equation}
which clearly approaches $1$ as $s \to \infty$.  
Hence, in multilayer settings with the assumption of independent components and additive costs, 
the advantage of IS diminishes with increasing size $s$,
any cost reduction achievable through independent sampling becomes asymptotically negligible. 
However, relaxing either of these two assumptions is enough to produce entirely different behaviors.
Moreover, distinct scenarios arise when error detection schemes are applied (see Section~\ref{sec : IS with error detection}).

\section{IS for imperfect quantum channels}\label{sec:IS_imperfect_channels}
In the previous section, 
we showed that average computational cost, 
for instance circuit depth, can be reduced by IS. 
This might suggest that the overall noise-induced error could change and eventually be decreased.
In this section we show that this intuition is misleading, 
as the bias of the final estimator is totally uneffected by IS. 
More generally,
in many practical scenarios, 
the quantum channel $\mathcal{E}_{\underline{\theta}}$ introduced in Eq.~\eqref{eq: state decomposition}
can only be approximated by a quantum channel \(\mathcal{N}_{\underline{\theta}}(\rho)\),
both because of algorithmic biases and hardware imperfections, including gate noise and decoherence.
Let
\begin{equation}
     \varepsilon_{\underline{\theta}}(\rho) := \mathcal{E}_{\underline{\theta}}(\rho) - \mathcal{N}_{\underline{\theta}}(\rho)\,,
\end{equation}
represent the error in the approximation of the state $\mathcal{E}_{\underline{\theta}}(\rho)$.
Using IS, 
the expected error in the prepared state is $\mathbb{E}_q[\varepsilon_{\underline{\theta}}(\rho)]$
and could, in principle, be reduced through a suitable choice of the sampling distribution $q(\underline{\theta})$. 
Nevertheless, 
such a reduction does not translate into a smaller bias in the final observable, 
since the latter is exactly compensated by the weighting factor $w(\underline{\theta})$. 
In other words, 
the final bias is fully determined by $p(\underline{\theta})$ 
and is not altered by the choice of the IS distribution,
as proved in the following theorem.
\begin{theorem}\label{thm:imperfect_quantum_channels}
Let 
\begin{align}
\mathcal{N}_{\underline{\theta}}(\rho)  = 
 \mathcal{E}_{\underline{\theta}}(\rho) 
- \varepsilon_{\underline{\theta}}(\rho)
\end{align}
be an approximate implementation of the channel $\mathcal{E}_{\underline{\theta}}$, 
applied with probability $p(\underline{\theta})$.
Let $q(\underline{\theta})$ denote the importance-sampling distribution. 
Then,
the bias introduced by the channel approximation in the 
importance-sampled estimator of any observable $O$ is independent of $q$ and is given by
\begin{align}
    \mathbb{E}_p\left[ \Tr\!\left[\, O \, \varepsilon_{\underline{\theta}}(\rho) \right] \right]
\end{align}
\end{theorem}

\begin{proof}
Define the ideal and noisy outcome probabilities for a POVM $\{\Pi_x\}_x$ as
\begin{align}
    r(x|\underline{\theta}) 
    &:= \Tr\!\left[\, \Pi_x\, \mathcal{E}_{\underline{\theta}}(\rho) \right], \\
    r'(x|\underline{\theta}) 
    &:= \Tr\!\left[\, \Pi_x\, \mathcal{N}_{\underline{\theta}}(\rho) \right]
\end{align}
respectively.
Replacing $\mathcal{E}_{\underline{\theta}}$ with $\mathcal{N}_{\underline{\theta}}$ 
leads to an estimator whose expectation value acquires a bias equal to
\begin{align}
    \mathbb{E}_{q r'}[w(\underline{\theta}) x] - \mathbb{E}_{p r}[x] &= \mathbb{E}_{p r'}[ x] - \mathbb{E}_{p r}[x] \nonumber \\
         &= \mathbb{E}_{p (r'-r)}[x]
\end{align}
where we used $p(\underline{\theta}) = q(\underline{\theta}) w(\underline{\theta})$. 
This shows the bias is $q$-independent, concluding the proof.

\end{proof}
In summary, while IS effectively reduces computational resources, 
it neither alleviates nor exacerbates the bias 
induced by imperfect state preparation or hardware noise. 
Ultimately, 
the expectation value of any observable 
is fully determined by the original distribution $p(\underline{\theta})$ 
and the implemented quantum channel $\mathcal{N}_{\underline{\theta}}$.

\section{IS with error detection}\label{sec : IS with error detection}

In this section, we consider a scenario in which a device, after completing a circuit execution, outputs a classical \emph{flag} 
$F \in \{\mathtt{ok}, \mathtt{err}\}$ indicating whether a detectable fault occurred. 
We refer to runs with $F=\mathtt{ok}$ as successful and runs with $F=\mathtt{err}$ as failed.
We assume that the flag is perfectly reliable, with no false positives or negatives. 
If an error occurs, two situations can arise. In some cases, the error can be corrected, 
allowing the final outcome to be treated as if no error occurred.
n these cases, the additional resources required for the error-correction procedure should be accounted for in the overall cost function.
In other cases, errors can be detected but not corrected. 
This situation is common in current quantum devices, for example, in trapped-ion systems, 
where certain faults can be reliably flagged but not actively corrected \cite{hilder2022fault, self2024protecting, ransford2025helios}. 
Alternatively, errors may be detected through violations of conservation laws or symmetries that are preserved by the ideal, error-free evolution~\cite{linke2018measuring,bonet2018low,mcardle2019error,cai2021quantum}. 

In the case of errors that can be detected but not corrected, various strategies can be employed.
First we describe how these situations can be handled without IS. Consider a randomized algorithm for estimating $\Tr[O \, \rho_\text{target}]$, 
where $\rho_\text{target}$ can be decomposed as in Eq.~\eqref{eq: state decomposition} using a probability distribution $p(\underline{\theta})$. 
The simplest procedure is to sample a value of $p(\underline{\theta})$, 
apply the corresponding parametrized quantum channel $\mathcal{E}_{\underline{\theta}}$ on the initial state $\rho$, 
and measure $O$. 
If the flag is $\mathtt{ok}$, the outcome is recorded; otherwise, the whole process is repeated with the same $\utheta$ until success. This is necessary in order for the samples $\utheta$ to come from the target distribution $p(\utheta)$ and the final state prepared be $\rho_\text{target}$.
This in turn implies that the expected cost to record one successful outcome depends on both the cost associated with $\mathcal{E}_{\underline{\theta}}$ 
and the success probability $f(\underline{\theta})>0$ of the flag indicating $\mathtt{ok}$.
Since we are not using IS at this stage, the expected cost coincides with the net cost
and reads
\begin{align}
    NC_p = \mathbb{E}_p\left[\frac{c(\underline{\theta})}{f(\underline{\theta})}\right]\,, \label{eq:netCostLinftyWithP}
\end{align}
as we'll prove later in this section.
When $f(\underline{\theta})$ is small for some $\utheta$, the corresponding circuit needs to be run many times on average before it doesn't flag an error. This drives up the average cost of the algorithm. In response one may want to impose a limit $L$ on the number of repetitions after which a new $\utheta$ is sampled from $p(\utheta)$. 
However, this effectively alters the sampling distribution $p(\utheta)$ 
and IS is required to maintain an unbiased estimator.
In what follows we present two approaches that allow recording independent outcomes whose expectation values exactly correspond to $\Tr[O \, \rho_\text{target}]$, 
differing in how they handle unsuccessful runs. 
We then analyze these methods and derive the optimal IS distributions that minimize the net cost for each protocol for a given target estimation error.

\subsection{ZeroFill($L$)}\label{sec:RUSZeroFill}
The ZeroFill$(L)$ protocol proceeds through a cycle of three steps:
\begin{enumerate}
    \item \textbf{Parameter Sampling:}  
    Sample a random  \(\underline{\theta}\) according to the distribution \(q(\underline{\theta})\).
    \item \textbf{Circuit Execution and Detection Loop:}  
    For the sampled \(\underline{\theta}\), run the corresponding quantum circuit, incurring a cost \(c(\underline{\theta})\) and subject to an error-detection scheme.  
    The execution is repeated up to a maximum of \(L\) times.
    \label{step:loop}
    \begin{itemize}
        \item \textbf{Successful Run:}  
        With probability \(f(\underline{\theta})\), no error is detected.  
        In this case, measure the observable \(O\), multiply the resulting outcome by the corresponding weight \(w_\mathrm{Z}(\utheta;L)\) (see Eq.~\eqref{eq: weights Zerofill}), and record it.  
        Then proceed to sample the next \(\underline{\theta}\) from $q(\utheta)$.

        \item \textbf{Failed Run (Error Detected):}  
        With probability \(1 - f(\underline{\theta})\), an error is detected.  
        The current circuit run is discarded, and the execution is repeated using the same \(\underline{\theta}\).
    \end{itemize}

    \item \textbf{Maximum-Repetition Rejection:}  
    If all \(L\) executions associated with the fixed \(\underline{\theta}\) result in detected errors, a value of \(0\) is recorded (Zero-Fill), and the procedure moves on to sample the next \(\underline{\theta}\).
\end{enumerate}

The objective is to estimate $\mathbb{E}_{p}\!\left[O\,\mathcal{E}_{\underline{\theta}}(\rho)\,\right]$
by repeating the procedure many times
and using the average of the results as the estimator.
In order for the estimator to be unbiased the weight function has to be (see Appendix~\ref{appx:proofZeroFillDiscardL} Eq.~\eqref{eq:unbiasedwDetectError})
\begin{align}\label{eq: weights Zerofill}
    w_\mathrm{Z}(\utheta;L) = \frac{p(\utheta)}{q(\utheta)k(\utheta;L)}\,,
\end{align}
where we introduced $k(\utheta;L) := 1-(1-f(\utheta))^L$ as the probability of 
having at least one successful run, 
for further convenience.
Notice that for $L = 1$ one has
\begin{align}
    k(\utheta; 1) = f(\utheta)\,,
\end{align}
while, in the opposite regime
\begin{align}
    \lim_{L \mapsto\infty} k(\utheta; L) = 1.
\end{align}

\subsection{Discard($L$)}
Similarly to the previous algorithm, in the Discard$(L)$ protocol we allow at most $L$ unsuccessful circuit repetitions for each independent sample $\underline{\theta}$.  
However, unlike the previous protocol, if all $L$ attempts of step \ref{step:loop} fail, no outcome is recorded and a new sample is drawn.  
To maintain the unbiasedness of the estimator, 
the weight function 
which we here denote as $w_{\mathrm{D}}(\utheta;L)$ must satisfy (see Appendix~\ref{appx:proofZeroFillDiscardL}  Eq.~\eqref{eq:discard_unbiasedwDetectError})
\begin{align}\label{eq: weights Discard}
    w_\mathrm{D}(\underline{\theta}; L) = \mathbb{E}_q\left[k(\underline{\theta}; L)\right] \frac{p(\underline{\theta})}{q(\underline{\theta})\,k(\underline{\theta}; L)} \, ,
\end{align}
where 
\begin{align}\label{eq : weights proportionality factor}
    \mathbb{E}_q\left[k(\underline{\theta};L)\right] = \int d\underline{\theta}\, q(\underline{\theta})\, k(\underline{\theta}; L) \in [0,1]
\end{align}
Notice that the weight function in Eq.~\eqref{eq: weights Discard} is proportional to that of the ZeroFill($L$) protocol:
\begin{align}
    w_\mathrm{D}(\underline{\theta}; L) =\mathbb{E}_q\left[k(\underline{\theta}; L)\right]\, w_\mathrm{Z}(\underline{\theta}; L).
\end{align}
In particular, since $k(\underline{\theta}; L) \to 1$ as $L \to \infty$, then
\begin{align}
    \lim_{L \to \infty}\mathbb{E}_q\left[k(\underline{\theta}; L) \right] = 1,
\end{align}
which implies that $w_\mathrm{Z}$ and $w_\mathrm{D}$ become identical in this limit.

\subsection{Algorithm performance}\label{sec : algorithms performances}
The performance of the two preceding algorithms can be assessed in terms of the expected cost per recorded outcome and the corresponding variance. 
As shown in the following theorem, 
these quantities admit identical expressions for both algorithms
in terms of their respective weight functions $w_\mathrm{Z}(\utheta;L)$ or $w_\mathrm{D}(\utheta;L)$.

\begin{theorem}\label{thm:algorithms performancesZeroFillDiscardL}
    Consider a target distribution $p(\underline{\theta})$, a cost function $c(\underline{\theta})$, 
    and one of the schemes for handling errors as described above.
    Let $f(\underline{\theta})$ be the probability that the circuit associated with $\utheta$ will run without error and  
    $q(\underline{\theta})$ be the sampling distribution for $\utheta$.
    The expected cost for a single recorded outcome,
    expressed in terms of the weight function $w_\mathrm{Z,D}(\utheta;L)$,
    is
    \begin{equation}\label{eq: cost Zerofill}
        \mathbb{E}_p\!\left[\frac{c(\underline{\theta})}{w_\mathrm{Z,D}(\underline{\theta};L)\, f(\underline{\theta})}\right],
    \end{equation}
    while its worst-case variance, for all observables satisfying $\|O\|\leq 1$, is
    \begin{align}\label{eq: variance Zerofill}
        \mathbb{E}_p\!\left[w_\mathrm{Z,D}(\underline{\theta};L)\right].
    \end{align}
\end{theorem}
\begin{proof}
  A proof is provided in Appendix~\ref{appx:proofZeroFillDiscardL}.
\end{proof}

\noindent
Notice that,
although Eqs.~\eqref{eq: cost Zerofill} and~\eqref{eq: variance Zerofill} hold for both the ZeroFill($L$) and Discard($L$) protocols,
their actual performance differs because the corresponding weight functions are not identical. 
In particular, the expected cost of ZeroFill($L$) equals that of Discard($L$) multiplied by the factor $\mathbb{E}_q\left[k(\underline{\theta};L)\right]$. Conversely, the variance bound of ZeroFill($L$) is equal to that of Discard($L$) divided by $\mathbb{E}_q\left[k(\underline{\theta};L)\right]$.
Also,
we emphasize that Eq.~\eqref{eq: cost Zerofill} quantifies the expected cost to record a successful outcome. 
If, instead, one is only interested in the cost of a single run of the algorithm regardless of whether the outcome is successful and ultimately recorded,
it suffices to consider $\mathbb{E}_q[c(\underline{\theta})]$.

\begin{corollary}\label{cor:optimal_Zerofill_Discard}
For both the ZeroFill($L$) and Discard($L$) protocols, characterized by a success probability
\( f(\underline{\theta}) \) and a sampling distribution \( q(\underline{\theta}) \),
the net cost satisfies
\begin{align}\label{eq:NetCostZeroFill}
    NC_q \;\geq\;
    \left(
        \mathbb{E}_p\!\left[
            \sqrt{\frac{c(\underline{\theta})}{f(\underline{\theta})}}
        \right]
    \right)^2.
\end{align}
Equality holds if and only if the sampling distribution is given by
\begin{align}\label{eq: optimal q detection}
    q_L^*(\underline{\theta})
    \;\propto\;
    \frac{p(\underline{\theta})}{k(\underline{\theta};L)}
    \sqrt{\frac{f(\underline{\theta})}{c(\underline{\theta})}}\,.
\end{align}
\end{corollary}

\begin{proof}
We consider the net cost
\begin{align*}
    NC_q
    \;\geq\;
    \mathbb{E}_p\!\left[
        \frac{c(\underline{\theta})}
        {w_\mathrm{Z,D}(\underline{\theta};L)\, f(\underline{\theta})}
    \right]
    \cdot
    \mathbb{E}_p\!\left[
        w_\mathrm{Z,D}(\underline{\theta};L)
    \right],
\end{align*}
which is obtained by multiplying the expected cost in
Eq.~\eqref{eq: cost Zerofill} with the variance bound in
Eq.~\eqref{eq: variance Zerofill}.
Applying the Cauchy--Schwarz inequality yields
Eq.~\eqref{eq:NetCostZeroFill}, with equality if and only if
\begin{align}
    w_\mathrm{Z,D}(\underline{\theta};L)
    \;\propto\;
    \sqrt{\frac{c(\underline{\theta})}{f(\underline{\theta})}}\,.
\end{align}
Finally, since both the ZeroFill($L$) and Discard($L$) protocols satisfy
\begin{align}
    q(\underline{\theta})
    \;\propto\;
    \frac{p(\underline{\theta})}
    {k(\underline{\theta};L)\, w_\mathrm{Z,D}(\underline{\theta};L)},
\end{align}
the optimal sampling distribution in
Eq.~\eqref{eq: optimal q detection} follows.
\end{proof}
\noindent
Remarkably,
both the optimal net cost and the associated sampling distribution $q^*_L$ are identical for the ZeroFill($L$) and Discard($L$) algorithms.
Also, 
while the optimal sampling distribution $q_L^*(\underline{\theta})$ depends on $L$, 
the resulting optimal net cost does not.
In the limit $L \to \infty$
we have $k(\underline{\theta}; L) \to 1$ and 
$\mathbb{E}_q\!\left[k(\underline{\theta}; L)\right] \to 1$. 
In this regime, the weight function for both algorithms reduces to the standard IS form
\begin{align}
    w(\underline{\theta}) = \frac{p(\underline{\theta})}{q(\underline{\theta})}\,,
\end{align}
and from Theorem~\ref{thm:algorithms performancesZeroFillDiscardL} we recover that sampling directly from $p(\underline{\theta})$ yields a worst-case variance bound equal to $1$ and the expected cost in Eq.~\eqref{eq:netCostLinftyWithP} which is strictly larger than the optimal net cost achievable through IS for any positive cost function. 
Note that the optimal sampling distribution in the $L \to \infty$ limit follows from Corollary~\ref{cor:optimal_Zerofill_Discard} and is given by
\begin{align}
    q^*_\infty(\underline{\theta}) \propto p(\underline{\theta}) 
    \sqrt{\frac{f(\underline{\theta})}{c(\underline{\theta})}}\,.
\end{align}


\subsection{A Case Study: IID Parameter Model under Poisson Error Statistics}
In this section, we present an analytically tractable instance of the general framework 
under the  assumptions introduced in Section~\ref{sec:multilayer}. 
We consider a collection of parameters $\{\theta_1,\ldots,\theta_s\}$ modeled as independent and identically distributed random variables, and assume that the total cost function $c(\underline{\theta})$ is additive, with identical single-parameter contributions associated with each $\theta_i$.
We model the probability of observing no detectable errors through a Poissonian ansatz that depends on the total cost,
\begin{align}
    f(\underline{\theta}) = e^{-\lambda c(\underline{\theta})}, \qquad \lambda > 0 \, .
\end{align}
This modeling choice is appropriate,
e.g.,
in settings where the cost function faithfully reflects the total execution time of a quantum circuit 
and errors arise from stationary noise processes acting continuously on the hardware. 
Under these assumptions, 
the probability of remaining error-free decays exponentially with the accumulated cost, 
naturally leading to the Poissonian form above.

\begin{theorem} \label{thm:TheoremIIDLInfty}
Given a target distribution $p(\underline{\theta})$, 
let $q(\underline{\theta})$ be an arbitrary sampling distribution. Consider i.i.d. random variables $\underline{\theta} = (\theta_1,\ldots,\theta_s)$ with an additive cost
\begin{align}
    c(\underline{\theta}) = \sum_{i=1}^s c_1(\theta_i).
\end{align}
Let the success probability of the detectable-error mechanism be
\begin{align}
    f(\underline{\theta}) = e^{-\lambda c(\underline{\theta})},
\end{align}
for some $\lambda \ge 0$. Define
\begin{align}
    \mu_1 = \mathbb{E}_p[c_1(\theta)], 
    \qquad
    \sigma_1^2 = \mathrm{Var}_p[c_1(\theta)],
\end{align}
as the mean and variance of the single-variable cost function $c_1(\theta)$.
Using one of the protocols described in Section~\ref{sec : IS with error detection} with $L = \infty$,
the ratio between the net cost obtained under the optimal sampling distribution $q^*(\underline{\theta})$ 
(as characterized by Theorem~\ref{thm:TheoremIIDLInfty}) and that obtained under the original distribution $p(\underline{\theta})$ admits the following asymptotic expansion
\begin{align}
\frac{NC_{q^*}}{NC_p}(s)
    &= 
      \left(
      \frac{\mu_1+ \frac{\lambda \sigma_1^2}{2}}
           {\mu_1+ \lambda \sigma_1^2}
           \right)e^{-s\left(\lambda\sigma_1\over2\right)^2}\nonumber \\
           & \;\; \times
      \left[1-\frac{1}{4s}
      \left(\frac{\sigma_1}{\mu_1+ \frac{\lambda \sigma_1^2}{2}}\right)^2
      + O\!\big(s^{-3/2}\big) \right].
\end{align}
\end{theorem}
\begin{proof}
  A proof is provided in Appendix \ref{proof:TheoremIIDLInfty}.
\end{proof}
\noindent
In the no-error regime $\lambda = 0$,
\begin{align}
    \frac{NC_{q^*}}{NC_p}(s)
    \sim 1- \frac{1}{s}\,,
    \qquad s \to \infty ,
\end{align}
and the ratio converges to unity, indicating the absence of an asymptotic advantage from IS.
In contrast, for any $\lambda > 0$ the ratio of net cost decays exponentially with the system size. In particular,
\begin{align}
    \frac{NC_{q^*}}{NC_p}(s)
    \sim
    \exp\!\left[- s\left(\frac{\lambda \sigma_1}{2}\right)^2 \right],
    \qquad s \to \infty ,
\end{align}
so that the ratio vanishes exponentially whenever $\sigma_1 \neq 0$.
Moreover, the ratio $NC_{q^*}/NC_p(s)$ admits a unique global maximum at a finite value of $s$, given by
\begin{align}
    s^*
    &\approx \frac{1}{8}
    \left(\frac{\sigma_1^2}{\mu_1 + \frac{\lambda \sigma_1^2}{2}} \right)^2
    \left(
        1+ \sqrt{
            1+\frac{64}{\lambda^2\sigma_1^4}
            \left( \mu_1 + \frac{\lambda \sigma_1^2}{2}\right)^2
        }
    \right).
\end{align}
This value marks the size of $\utheta$ at which the relative net-cost reduction 
induced by IS is minimal before the exponential behavior dominates.

\section{Applications}\label{sec:applications}

In this section, we present representative applications that directly build on the results established above. Each application is treated in a separate subsection. For clarity, the subsections are self-contained and can be read independently and in arbitrary order.
Most examples concern scenarios in which the randomization is confined to a single block of operations within the quantum circuit ($s = 1$) and no error-detection scheme  is employed ($\lambda = 0$). The extension to multiple independently sampled blocks treated in Sec.~\ref{sec:multilayer}, as well as to settings with error detection, 
where Theorem~\ref{thm:TheoremIIDLInfty} applies,
is straightforward.

\subsection{\textsc{Qdrift}}
Given an Hamiltonian $H$,
the \emph{\textsc{Qdrift} protocol} \cite{campbell2018random} provides a stochastic method to approximate the time evolution of an inital quantum state $\rho$ under $H$ for a duration $t$.
This is achieved by \textsc{Qdrift} using a quantum channel $\mathcal{C}^s$
that is the composition of an integer number $s$ of identical channels $\mathcal{C}$,
each approximating
\begin{equation}
\mathcal{U}[\rho] = e^{-i \frac{t}{s} H} \rho e^{i \frac{t}{s} H} = e^{\frac{t}{s} \mathcal{L}}[\rho], 
\end{equation}
where 
\begin{equation}
\mathcal{L}[\rho] = -i[H, \rho]
\end{equation}
is the Liouvillian.
Let
\begin{equation}
H = \sum_j H_j,
\end{equation}
be a decomposition such that each term $H_j$ is associated with a unitary operator $V_j = e^{-i H_j t_j}$ that can be implemented on quantum hardware
with a cost $c(j)$. 
The single-step \textsc{Qdrift} channel is defined by
\begin{align}
\mathcal{C}[\rho] &= \sum_j p(j) e^{-i \frac{t}{s p(j)} H_j} \,\rho\, e^{i \frac{t}{s p(j)} H_j} \\
&= \sum_j p(j) e^{\frac{t}{s p(j)} \mathcal{L}_j}[\rho]\,,
\end{align}
where 
\begin{equation}
\mathcal{L}_j[\rho] = -i[H_j, \rho],
\end{equation}
and
\begin{equation}\label{eq:single_step_distribution_QDrift}
p(j) = \frac{\norm{H_j}}{\lambda}\,, \qquad \lambda = \sum_j \norm{H_j}\,.
\end{equation}
Approximating the $\mathcal{U}$ channel according to the \textsc{Qdrift} 
protocol in order to estimate the expectation value of an observable $O$ yields an error bounded by
\begin{align}\label{eq:qdrift_bias}
    \big|\Tr[O \mathcal{C}[\rho]] - \Tr[O \mathcal{U}[\rho]]\big| \leq \norm{O}\left(\frac{2 t \lambda}{s}\right)^2 e^{\frac{2 t \lambda}{s}}\,.
\end{align}
Following the framework of IS 
from Sec.~\ref{sec: General Discussion}, 
we introduce the more general channel
\begin{align}
\mathcal{C}_q[\rho] &= \sum_j q(j) e^{\frac{t }{s p_j} \mathcal{L}_j}[\rho]\,,
\end{align}
where $q(j)$ is an alternative sampling distribution,
and simultaneously multiply the observable measured at the end of the circuit
by the corresponding weights  $w(j) = p(j)/q(j)$.
Notice that the conventional \textsc{Qdrift} protocol
is recovered with the choice $q(j) = p(j)$.
Applying the optimality condition from Theorem~\ref{thm:optSampling}, 
the net-cost minimizing distribution is
\begin{equation}\label{eq:optimal_sampling_Qdrift}
q^*(j) \propto \frac{p(j)}{\sqrt{c(j)}}.
\end{equation}
This choice minimizes the expected total cost of the single-step \textsc{Qdrift} channel
while maintaining the bias in Eq.\eqref{eq:qdrift_bias} unchanged,
according to Theorem~\ref{thm:imperfect_quantum_channels}.

\subsubsection{\texorpdfstring{He$_2$}{He2} example}
\label{sec:qdrift_he2_example}
We next illustrate the advantage of the above result through a case study of the helium dimer Hamiltonian. The electronic-structure Hamiltonian of He$_2$ is mapped to qubits via a standard fermion-to-qubit transformation (e.g., Jordan–Wigner or Bravyi–Kitaev), followed by basis truncation.  
After mapping, the Hamiltonian can be expressed as a real linear combination of Pauli strings on $n=8$ qubits,
\begin{equation}
H = \sum_{j=1}^{M} h_j P_j,
\;\;
h_j\in\mathbb{R},\;\;
P_j\in\{I,X,Y,Z\}^{\otimes 8}\,.
\label{eq:pauli_decomp}
\end{equation}
In this example, we use the molecular Hamiltonian coefficients $\{h_j\}$ provided by \cite{pennylane}.
To incorporate cost-aware IS, 
we assign each Pauli word $P_j$ a circuit cost $c(j)$ that approximates the number of entangling gates required to implement $e^{-i\theta P_j}$.  
Let
\begin{equation}
S(P_j) \;:=\; \big|\{k \in \{0,\dots,7\} : (P_j)_k \neq I\}\big|.
\label{eq:pauli_weight}
\end{equation}
be the number of qubits on which the Pauli string $P_j$ is acting non trivially.
Using the standard parity-computation circuit (basis change to $Z^{\otimes S}$, a CNOT ladder/tree to compute parity, a single $R_z(2\theta)$, and uncomputation), an all-to-all connectivity cost model yields
\begin{equation}
c(j) \;:=\; 2\big(S(P_j)-1\big)
\label{eq:cnot_cost}
\end{equation}
CNOT operations. This model captures the dominant scaling of entangling resources with Pauli weight: one-local terms (e.g.\ $Z_k$) require no CNOTs, while two-local terms (e.g.\ $Z_a Z_b$) require $2$ CNOTs, and so on.
Such a cost model gives
\begin{align}
    \mathbb{E}_p[c(j)] = 4.4735\,, \qquad \left(\mathbb{E}_p[\sqrt{c(j)}]\right)^2 = 2.8468\,,
\end{align}
where $p$ is given by Eq.~\eqref{eq:single_step_distribution_QDrift}.
As a result
\begin{align}
    \frac{NC_{q^*}}{NC_p} = 0.6364.
\end{align}
showing that the optimal IS scheme gives an improvement of $36.36\%$ compared the standard \textsc{Qdrift} in terms of Net Cost formulation for the single time step.

\subsubsection{Comparison with Related Approaches}
The authors of~\cite{kiss2023importance} explore a related idea to improve the performance of \textsc{Qdrift}, referring to their approach as IS. In the original \textsc{Qdrift}~\cite{campbell2018random}, Hamiltonian terms are sampled with probability proportional to their operator norm, and the time evolution corresponding to each term is performed for a duration inversely proportional to that probability. In~\cite{kiss2023importance}, a different sampling distribution is proposed, and the evolution times are adjusted accordingly to ensure that the quantum channel implemented by the randomized protocol remains unchanged. This is achieved by multiplying the standard \textsc{Qdrift} evolution times by a weight function defined in~\cite{kiss2023importance}.

This approach differs from ours in two fundamental ways. First, the modified evolution times correspond to unitary channels that are generally different from those used in the original \textsc{Qdrift} algorithm, whereas our IS method samples only channels that already appear in the original randomized protocol. Second, the “weights” in~\cite{kiss2023importance} modify the quantum channel itself rather than multiplying observed quantities, as in our IS formulation. From our perspective, the method of~\cite{kiss2023importance} can be interpreted as an alternative randomized protocol to \textsc{Qdrift}, in which a different set of channels is sampled according to a modified distribution to approximate the target time-evolution unitary. In fact, we could directly apply our IS method to the algorithm of~\cite{kiss2023importance}.

The conceptual perspective also differs. In~\cite{kiss2023importance}, the per-run cost is reduced using a heuristic modified distribution, but the variance of the resulting estimator increases such that the net cost of computation is ultimately higher. In contrast, our approach explicitly aims to minimize the final net cost by providing an analytically optimal sampling distribution. Additionally, as discussed in Sec.~\ref{sec:IS_imperfect_channels}, our IS framework preserves the effects of noise and cannot mitigate them, whereas the method of~\cite{kiss2023importance} modifies the implemented quantum channels and can, in principle, change the effect of noise.

\subsection{Dephasing channels}\label{sec : Eigenpath Transversal}
A direct application of the results developed in this work is the realization of a dephasing quantum channel in the eigenbasis of a given operator $H$.
Such channels are a key primitive in Zeno-based state stabilization and 
adiabatic state preparation \cite{BoixoKnillSomma2009, CunninghamRoland2024}.
More precisely, 
in the context of state preparation,
the objective is to implement a quantum channel of the form \cite{BoixoKnillSomma2009}
\begin{align}\label{eq:Zeno_decoherence_channel}
    \left[\mathcal{D}_l\right](\rho)
    =
    P_l \rho P_l
    +
    \mathcal{M}_l\!\left(Q_l \rho Q_l\right),
\end{align}
where
\begin{align}
    P_l = \ket{\psi_l}\bra{\psi_l}
\end{align}
denotes the projector onto the instantaneous $l$-th eigenstate of $H$,
\begin{align}
    Q_l = \mathds{1} - P_l
\end{align}
is the corresponding orthogonal projector onto the complementary subspace, and $\mathcal{M}_l$ represents an arbitrary quantum channel acting within that subspace.
A practical approach to achieve this transformation consists in exploiting randomized unitary dynamics that effectively approximate the action of the dephasing channel $\mathcal{D}_l$.
Specifically, 
one samples a time parameter $t$ from a classical probability distribution $p(t)$ 
and applies the unitary evolution
\begin{align}\label{eq:eigenpath_transversal_evolution}
    U(t) = e^{-i H t},
\end{align}
thereby realizing the quantum channel
\begin{align}
    \mathcal{E}^p(\rho)
    &=
    \int dt\, p(t)\, U(t)\rho U^\dagger(t) \\
    &=: \int dt\, p(t)\, \mathcal{E}_t(\rho).
\end{align}
This expression is formally analogous to Eq.~\eqref{eq: state decomposition}, with the variable $t$ playing the role of $\utheta$.
The discrepancy between the ideal channel $\mathcal{D}_l$ and its randomized approximation $\mathcal{E}^p$ can be quantified in trace norm and was bounded in~\cite{BoixoKnillSomma2009} as
\begin{align}\label{eq: Zeno_error_bound}
    &\norm{\left[\mathcal{D}_l-\mathcal{E}^p\right](\rho)}_\mathrm{tr}
    \leq\sup_{j \neq l}
    \hat{p}\!\left(\abs{E_l-E_j}\right),
\end{align}
where $E_j$ is the $j$-th eigenvalue of $H$ and
$\hat{p}(\omega)$ is the Fourier transform of the distribution $p(t)$,
\begin{align}
    \hat{p}(\omega) := \int dt\, p(t)\, e^{-i \omega t}.
\end{align}
The approximation error in Eq.~\eqref{eq: Zeno_error_bound} vanishes provided that $\hat{p}$ satisfies
\begin{align}\label{eq:characteristic_function_condition}
    \hat{p}(\omega) = 0
    \qquad
    \forall\, \omega \geq \Delta > 0,
\end{align}
where $\Delta$ denotes lower-bound on the minimum energy gap separating the $l$-th eigenvalue from the remainder of the spectrum.
This condition can be satisfied only when the eigenvalue $E_l$ remains spectrally isolated, highlighting the essential role of a finite gap.
Crucially, the choice of $p(t)$ that makes the approximation error vanish in not unique.

In the following, 
we present an original contribution that addresses scenarios where the implementation of the unitary evolution $U(t)$
incurs a cost described by a generic function $c(t)$, 
which we assume depends solely on $|t|$.
Our objective is to determine the probability distribution $p(t)$ 
that minimizes the expected net cost of implementing the quantum channel $\mathcal{E}^p$ while respecting Eq.~\eqref{eq:characteristic_function_condition}, 
both with and without IS.
These two settings generally admit different optimal solutions, 
as they correspond to the minimization of distinct quantities, namely
\begin{align}
    NC_p = \mathbb{E}_p[c(t)]
\end{align}
in the absence of IS, and
\begin{align}
    NC_{q^*} = \left(\mathbb{E}_p[c^{1/2}(t)]\right)^2
\end{align}
when optimal IS is employed, as established in Theorem~\ref{thm:optSampling}.
To treat both cases within a unified framework, we focus on the minimization of
\begin{align}
  \mathbb{E}_p[z(\abs{t})]\,, 
\end{align}
where $z$ can be any function of $\abs{t}$.
\begin{widetext}
\noindent
In order to report the solution of the optimization problem in the following theorem,
we introduce the following two matrices,
\begin{align}\label{eq:A_B_definition}
\mathds{A}_{nm}(z,\Delta) &:= 
16\pi \Delta \Big(n-\tfrac{1}{2}\Big)\Big(m-\tfrac{1}{2}\Big)
\int dt\,
\frac{z(\abs{t}) \cos^2(\Delta t/2)}
{\bigl(\Delta^2t^2-(2\pi)^2(n-\tfrac{1}{2})^2\bigr)
 \bigl(\Delta^2t^2-(2\pi)^2(m-\tfrac{1}{2})^2\bigr)}, \\[2mm]
\mathds{B}_{nm}(z,\Delta) &:= 
16\pi \Delta n\,m
\int dt\,
\frac{z(\abs{t}) \sin^2(\Delta t/2)}
{\bigl(\Delta^2t^2-(2\pi)^2n^2\bigr)
 \bigl(\Delta^2t^2-(2\pi)^2m^2\bigr)}\,,
\end{align}
where the indices $n,m$ run over positive integers.  
We then collect these matrices into the block-diagonal matrix
\begin{align}\label{eq:V_matrix_definition}
\mathds{V}(z,\Delta) :=
\begin{bmatrix}
\mathds{A}(z,\Delta) & \mathbf{0} \\[1mm]
\mathbf{0} & \mathds{B}(z,\Delta)
\end{bmatrix}\,.
\end{align}
Accordingly, we denote the vectors on which $\mathds{V}(z,\Delta)$ acts as
\begin{align}
\underline{v} := (\underline{a},\underline{b}),
\end{align}
where $\underline{a}$ and $\underline{b}$ collect the coefficients associated with the cosine and sine blocks, respectively.
Notice that, since $\mathds{V}(z,\Delta)$ is block-diagonal, 
its basis of eigenstates can be chosen such that the support lies entirely in one of the two blocks, namely
\begin{align}
\underline{\tilde{v}} = [\underline{\tilde{a}},0]
\qquad \text{or} \qquad
\underline{\tilde{v}} = [0,\underline{\tilde{b}}].
\end{align}
The following theorem shows that,
within the framework described above, the problem of identifying the optimal distribution 
minimizing $\mathbb{E}_p[z(\abs{t})]$ can be reduced to an eigenvalue problem for the $\mathds{V}(z,\Delta)$ matrix.

\begin{theorem}[Optimal Band-Limited Randomized Sampling]
\label{thm:Optimal_Band_limited_Randomized_Sampling}
Let $\underline{\tilde{v}}$ be the normalized eigenvector of $\mathds{V}(z,\Delta)$ associated with its smallest eigenvalue $\nu$.  
Then 
\begin{align}
p_{\rm opt}(t) &= 
\sum_{n,m=1}^\infty
\frac{16\pi\Delta (n-\tfrac{1}{2})(m-\tfrac{1}{2})\,
\tilde{a}_n^\ast \tilde{a}_m}
{\bigl(\Delta^2t^2-(2\pi)^2(n-\tfrac{1}{2})^2\bigr)
 \bigl(\Delta^2t^2-(2\pi)^2(m-\tfrac{1}{2})^2\bigr)}
\cos^2(\Delta t/2) \nonumber\\
&\quad +
\sum_{n,m=1}^\infty
\frac{16\pi\Delta \,n\,m\,
\tilde{b}_n^\ast \tilde{b}_m}
{\bigl(\Delta^2t^2-(2\pi)^2 n^2\bigr)
 \bigl(\Delta^2t^2-(2\pi)^2 m^2\bigr)}
\sin^2(\Delta t/2)\,,
\end{align}
where exactly one of the coefficient vectors $\underline{\tilde{a}}$ or $\underline{\tilde{b}}$ is nonzero,  
is the probability density whose Fourier transform vanishes outside the interval $[-\Delta,\Delta]$ which minimizes $\mathbb{E}_p[z(\abs{t})]$.
The minimum value of $\mathbb{E}_p[z(\abs{t})]$ achieved is $\nu$.
\end{theorem}
\end{widetext}
\begin{figure*}[t]
    \centering
    \makebox[\textwidth][c]{%
        \subfigure[Dependence of the smallest eigenvalue of $\mathds{V}(\abs{t},\Delta)$ on the the series' truncation $n_\mathrm{max}$.]{
            \includegraphics[width=0.48\textwidth]{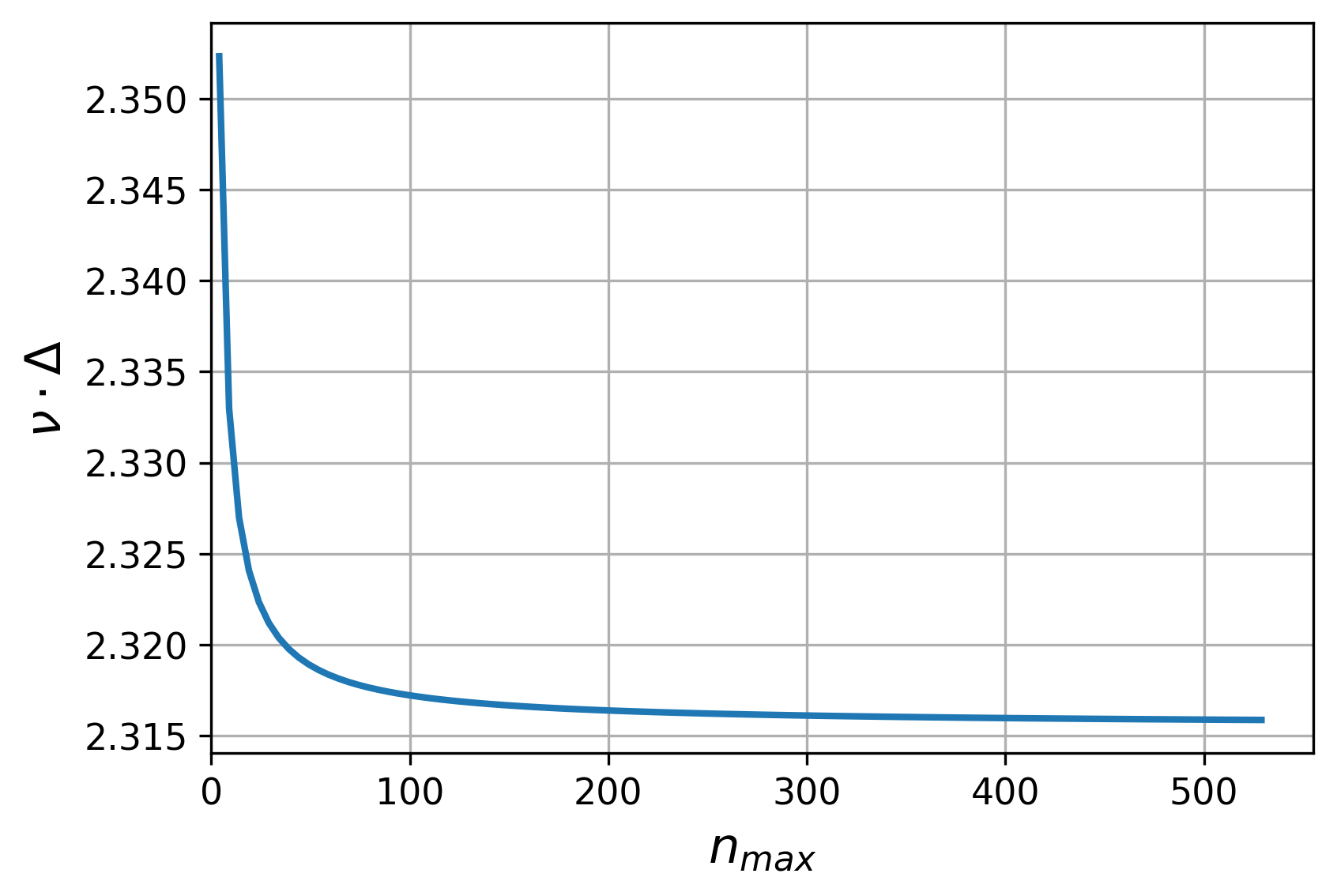}
        }\hspace{0.04\textwidth}
        \subfigure[Coefficients of the cosine-block for the eigenvector corresponding to the smallest eigenvalue of $\mathds{V}(\abs{t},\Delta)$.
            The values of the coefficients are obtained from the numerical diagonalization of the truncated matrix $\mathds{A}(\abs{t},\Delta)$ for $n = 1, 2, \dots, 529$.]{
            \includegraphics[width=0.48\textwidth]{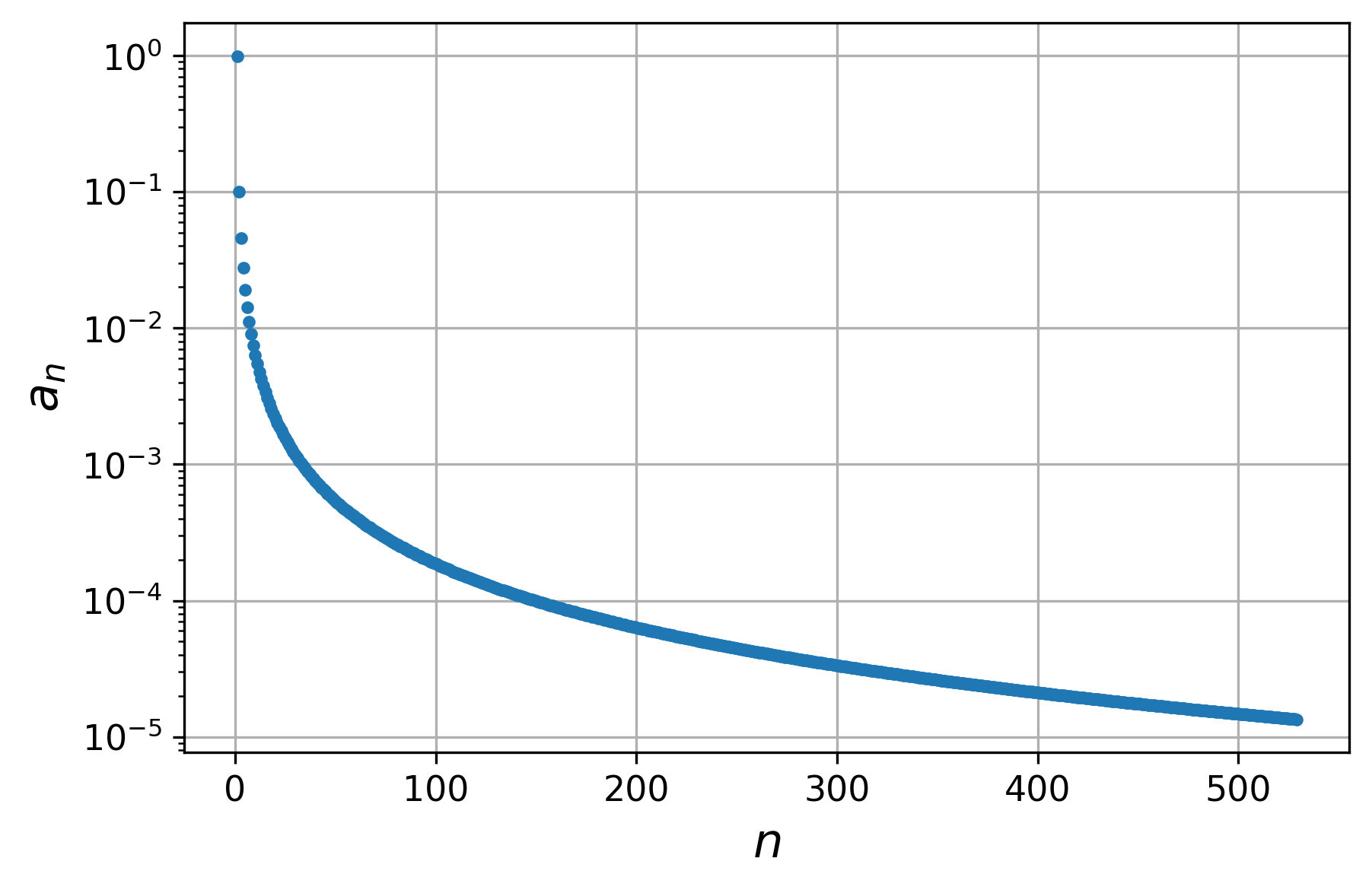}
        }
    }
    \caption{}
    \label{fig:optimal_eigenvector_A}
\end{figure*}
We report the proof of the theorem in Section~\ref{sec:proof_Optimal_Band_limited_Randomized_Sampling}.
The application of Theorem~\ref{thm:Optimal_Band_limited_Randomized_Sampling} 
to the minimization of 
\begin{align}
    \mathbb{E}_p[c^{\alpha}(t)]\,, \quad \alpha = \Big{\{}\frac{1}{2}, \,1\Big{\}}
\end{align}
is straightforward.  
As a concrete example, suppose that after sampling a value of $t$, the action of $U(t)$
is approximated using first-order  Trotter-Suziki formula up to a target precision (higher order formulas can be analyzed similarly). 
Since the error of the first-order Trotter-Suzuki formula scales as $\mathcal{O}(t^2/N)$, 
where $N$ is the number of Trotter steps, 
the final depth of the quantum circuit required to implement the evolution scales as $t^2$. 
If we take the circuit depth as the cost function, 
this is exactly of the form $c(t) \propto t^2$.  
To identify the optimal sampling distribution, we solve the eigenvalue problem
\begin{align}
    \mathds{V}(\abs{t}^{2 \alpha}, \Delta) \, \underline{\tilde{v}} = \nu \, \underline{\tilde{v}},
\end{align}
where $\nu$ is the smallest eigenvalue of $\mathds{V}(\abs{t}^{2 \alpha}, \Delta)$.  
For $\alpha = 1$, it can be analytically shown that (see Appendix~\ref{appx:A_B_computation})
\begin{align}
    \mathds{A}_{nm}(t^2, \Delta) &= \left(\frac{2\pi}{\Delta}\right)^2 \Big(n-\tfrac{1}{2}\Big)^2 \delta_{nm}, \\
    \mathds{B}_{nm}(t^2, \Delta) &= \left(\frac{2\pi}{\Delta}\right)^2 n^2 \,\delta_{nm}.
\end{align}
Hence, $\mathds{V}(\abs{t}^{2 \alpha}, \Delta)$ is diagonal, and the normalized eigenvector corresponding to its smallest eigenvalue has entries
\begin{align}
\tilde{a}_1 = 1, \quad \tilde{a}_{n > 1} = 0, \quad \tilde{b}_n = 0.
\end{align}
As a result of Theorem~\ref{thm:Optimal_Band_limited_Randomized_Sampling}, 
the optimal distribution to sample from (without performing IS) is
\begin{align}\label{eq:optimal_p_no_IS}
p_{\rm opt}(t) &= 
4 \pi \, \Delta\,
\frac{\cos^2\left(\frac{\Delta t}{2}\right)}
{\bigl(\Delta^2t^2 - \pi^2\bigr)^2}\,,
\end{align}
which allows to achieve a net-cost 
\begin{align}
    NC_{p_\mathrm{opt}} =\nu =  \left( \frac{\pi}{\Delta}\right)^2\,.
\end{align}

We now turn to the case $\alpha = 1/2$, corresponding to the optimization of the expected cost under IS.
In this regime, the objective function takes the form $\mathbb{E}_p[\abs{t}]$.
We derived closed-form expressions for the matrix elements of $\mathds{A}(\abs{t},\Delta)$ and $\mathds{B}(\abs{t},\Delta)$, reported in Appendix~\ref{appx:A_B_computation}, using symbolic computation.
In contrast with the case $\alpha = 1$, these matrices are no longer diagonal.
As a result, the optimal sampling distribution cannot be identified by inspection of a single matrix element, but instead requires the solution of a genuine eigenvalue problem.
To this end, we numerically diagonalize truncated versions of the matrices $\mathds{A}(\abs{t},\Delta)$ and $\mathds{B}(\abs{t},\Delta)$, restricting the indices to the finite set
\begin{align}
n,m \in {0,1,\dots,n_{\mathrm{max}}}.
\end{align}
The truncation is justified by the rapid decay of the matrix elements with increasing indices, and convergence is verified a posteriori by increasing $n_{\mathrm{max}}$.
Remarkably, for all truncation sizes considered, we find that the eigenvector associated with the smallest eigenvalue of $\mathds{V}(\abs{t},\Delta)$ has support entirely within the cosine block, namely
\begin{align}
\tilde{b}_n = 0 \qquad \forall, n,
\end{align}
while the coefficients ${\tilde{a}_n}$ are nontrivial.
This implies that the optimal distribution,
that we denote with $p^*_{\rm opt}(t)$,
is entirely supported on the family of cosine-modulated basis functions appearing in Theorem~\ref{thm:Optimal_Band_limited_Randomized_Sampling}.
By increasing the truncation to $n_{\mathrm{max}} = 529$, 
we obtain a converged estimate for the smallest eigenvalue, $\nu \cdot \Delta = 2.3159$, 
which is compatible with the performance $2.3160$ reported in~\cite{sanders2020compilation}, 
achieved using an analogous truncation of a series expansion based on a different (polynomial rather than trigonometric) functional basis.
Unlike the distribution in Eq.~\eqref{eq:optimal_p_no_IS},
which is optimal in the absence of IS, the distribution $p^*_{\rm opt}(t)$ is optimal specifically in the importance-sampling setting.
That is, it does not minimize the expected cost directly, but, when combined with IS according to Theorem~\ref{thm:optSampling}, 
where the sampling function is 
\begin{align}
    q^*_{\rm opt}(t) \propto \frac{p^*_{\rm opt}(t)}{\abs{t}}\; ,
\end{align}
allows to achieve the net cost
\begin{align}
NC_{q^*_\mathrm{opt}} = \nu^2 \simeq \left(\frac{2.32}{\Delta}\right)^2.
\end{align}
The coefficients ${a_n}$ of the normalized eigenvector associated with the smallest eigenvalue are shown in Fig.~\ref{fig:optimal_eigenvector_A}.

Finally, it is instructive to compare the performance of the optimal distribution $p_{\mathrm{opt}}$ in the absence of IS with that of the optimal distribution $p^*_{\mathrm{opt}}$ when IS is employed.
Using the results derived above
we obtain
\begin{align}
\frac{NC_{q^*_{\mathrm{opt}}}}{NC_{p_{\mathrm{opt}}}}
\simeq
\left(\frac{2.32}{\pi}\right)^2
\simeq 0.54\,.
\end{align}
This ratio provides a quantitative measure of the advantage offered by IS in the present setting, demonstrating that the optimal importance-sampled strategy reduces the net cost by almost a factor of two relative to the best achievable strategy without IS.

Dephasing in the eigenbasis of a unitary can be achieved in a similar fashion by acting with random powers of the unitary~\cite{BoixoKnillSomma2009}. In this case the probability distribution is over integers that correspond to the times the unitary is applied. A near optimal choice for this distribution is provided in Eq.(17) of Ref.~\cite{jennings2025randomized}. It is possible to implement IS with that distribution or derive the discrete distribution that is optimal when used with IS, following the same steps for continuous distributions above.
 
One of the most important applications of the dephasing channel we discussed in this section is adiabatic state preparation~\cite{BoixoKnillSomma2009}. For this application multiple, say $s$, dephasing channels are applied sequentially, each with respect to a different Hamiltonian along the interpolation between a simple initial Hamiltonian and a target Hamiltonian. As shown in Theorem~\ref{thm:TheoremIIDLInfty}, in the absence of error-detection schemes the improvement provided by IS decreases as $1/s$. 
In this regime, 
further developments are required to make IS more effective 
for typical adiabatic state-preparation applications, 
for instance by extending the search for optimal sampling distributions that incorporate correlations between the samples associated with different channels.
In contrast, when an error-detection scheme is available, the improvement grows exponentially with $s$.

\subsection{Mixed States and Composite Observables}
In this section, we consider two standard tasks in quantum computing: i) the preparation of mixed states and ii) the estimation of observables expressed as linear combinations of operators that are not jointly measurable. 
These tasks are conceptually distinct and may arise independently or concurrently within a given algorithm. 
Nevertheless, we analyze them within a common framework, as they give rise to closely related mathematical structures.

We begin by considering the estimation of a physical observable $O$ on a mixed quantum state
\begin{align}
    \rho_\mathrm{target} = \sum_n p_n \ket{\psi_n}\bra{\psi_n}.
\end{align}
A straightforward strategy consists of executing a collection of quantum circuits in which the register is initialized in the pure state $\ket{\psi_n}$ with probability $p_n$, followed by a measurement of $O$. The resulting measurement outcomes are then averaged to estimate $\Tr[O\rho_\text{target}]$. In practice, however, the cost associated with preparing the states $\ket{\psi_n}$ may vary significantly with $n$, making this naive approach potentially inefficient.

A closely related situation arises when the observable $O$ cannot be measured directly, but instead admits a decomposition into a linear combination of directly measurable operators $\{O_n\}$,
\begin{align}
    O = \sum_n a_n O_n, \qquad  a_n \in \mathds{R}.
\end{align}
A basic approach is to independently estimate each expectation value $\Tr[O_n \rho_\text{target}]$ and subsequently combine the results using classical post-processing. A more refined strategy introduces a sampling distribution
\begin{align}
    p_n := \frac{|a_n|}{\sum_m |a_m|},
\end{align}
from which the index $n$ is drawn randomly. Each experimental run then consists of measuring $O_n$ with probability $p_n$, and rescaling the measurement outcome by the factor $\sum_m |a_m|$ to obtain an unbiased estimator of $\Tr[O\rho_\text{target}]$.
Our interest lies in regimes where the cost of implementing the measurement associated with $O_n$ is strongly dependent on $n$. This typically occurs when the measurement requires a basis change whose circuit depth or gate count varies across different terms. A prominent example is provided by Variational Quantum Eigensolvers, e.g. when the objective function corresponds to a many-body Hamiltonian decomposed into a sum of Pauli strings. In this case, the cost of implementing each Pauli measurement depends sensitively on the hardware connectivity and the native gate set.

In both scenarios described above, the estimation procedure ultimately reduces to executing randomized quantum circuits sampled according to a classical probability distribution $p_n$. When the measurement costs $c_n$ associated with different circuits are nonuniform, it becomes advantageous to adjust the sampling distribution. This approach is adopted, for instance, in~\cite{arrasmith2020operator}, although not in a way that is optimal with respect to the net cost.
Theorem~\ref{thm:optSampling} guarantees that sampling according to
\begin{align}
    q_n^* \propto \frac{p_n}{\sqrt{c_n}}
\end{align}
minimizes the net cost required to achieve a fixed target precision for the estimator,
provided that the measurement outcomes are appropriately reweighted to preserve unbiasedness.

\subsection{Classical Shadows}

In what follows, we consider the Classical Shadows protocol~\citep{HuangKuengPreskill2020} applied to a system of $n$ qubits, 
described by an Hilbert space of dimension $d = 2^n$. 
The goal of the protocol is to efficiently estimate expectation values 
of a large family of observables on a quantum state $\rho_\text{target}$ using a limited number of randomized measurements. 
The procedure consists of two main steps: first, randomly sampled unitaries $U$ are applied to 
$\rho_\text{target}$, 
followed by projective measurements in the computational basis; 
second, a classical post-processing stage reconstructs unbiased estimators of the desired observables.
We focus on the case where the Classical Shadows ensemble consists of the $n$-qubit Clifford group $\mathcal C_n$, 
sampled uniformly according to $p(U) = 1 / |\mathcal C_n|$. 
Denoting by $c(U) > 0$ the implementation cost of a Clifford unitary $U \in \mathcal C_n$, 
we show how IS can reduce the net cost of the protocol while preserving its statistical accuracy.

\paragraph{Shadow channel and its inverse (uniform Clifford).}
Let $\mathcal{M}_p$ denote the channel that describes the first step of the Classical Shadows protocol:
\begin{align}
\mathcal{M}_p(X)
&= \mathbb{E}_{p} \left[\sum_{b\in\{0,1\}^n}
\langle b|UXU^\dagger|b\rangle\;U^\dagger|b\rangle\!\langle b|U \nonumber\right] \\
&= \frac{X+\operatorname{Tr}(X)\,I}{d+1}.
\end{align}
Its inverse is given by
\begin{align}
\mathcal{M}_p^{-1}(Y)
= (d+1)\,Y - \operatorname{Tr}(Y)\,I.
\end{align}
For a single measurement shot with $U\sim p$ and outcome $b$, 
one can define the corresponding classical snapshot as
\begin{align}
\rho(U,b)
&= \mathcal{M}_p^{-1}\!\left(U^\dagger|b\rangle\!\langle b|U\right)\nonumber \\
&= (d+1)\,U^\dagger|b\rangle\!\langle b|U - I.
\end{align}
Importantly, this inversion is not directly performed on the quantum device.
Indeed, $\mathcal{M}_p^{-1}$ is not completely positive and thus cannot be implemented as a quantum channel. 
Classical shadows rely on the fact that,
for any observable $O$,
the estimator $\Tr[\rho(U,b)\, O]$ is an unbiased estimator,
namely
\begin{align}
    \mathbb{E}_p \left[ \sum_{b\in\{0,1\}^n} p\left( b|U\right)\Tr[\rho(U,b) O]\right] = \Tr[\rho_\text{target} O] \,,
\end{align}
where we introduced the conditional probability 
\begin{align}
    p(b|U):= \langle b|U \rho_\text{target} U^\dagger|b\rangle\,. 
\end{align}
Furthermore, computing $\Tr[\rho(U,b) \,O]$ is classically efficient, 
making the protocol scalable to large systems.
In order to move to the IS setting, we rewrite the channel $\mathcal{M}_p$ as an expectation with respect to an arbitrary probability density function $q(U)$,
\begin{align}
\mathcal{M}_p(X)
= \mathbb{E}_q\left[
w(U)\sum_{b\in{0,1}^n}
\langle b|U X U^\dagger|b\rangle \;
U^\dagger |b\rangle\langle b| U
\right],
\end{align}
where the weight
\begin{align}
w(U) = \frac{p(U)}{q(U)},
\end{align}
accounts for the change of sampling measure.
This representation shows that the channel $\mathcal{M}_p$ can be reproduced without bias by sampling unitaries $U$ according to $q(U)$ rather than $p(U)$, provided that each measurement outcome is reweighted by the corresponding factor $w(U)$. The protocol therefore proceeds exactly as in the conventional classical shadows scheme, except for this additional classical reweighting step.
Accordingly, for a single measurement outcome $(U,b)$, the estimator of the observable $O$ is defined as
\begin{align}
\hat{O} := \Tr\left[\rho(U,b) w(U) O\right]\,.
\end{align}
This estimator is unbiased, since
\begin{align}
    \mathbb{E}_q \left[ \sum_{b\in\{0,1\}^n} p\left( b|U\right)\hat{O}\right] = \Tr[\rho_\text{target} O] \,.
\end{align}
If we denote the cost 
of running the quantum circuit for the unitary $U$ with $c(U)$,
the expected cost per sample is given by
\begin{align}
    \mathbb{E}_q\left[c(U)\right] = \mathbb{E}_p\left[\frac{c(U)}{w(u)}\right]\,,
\end{align}
whereas the variance of the final estimator can be bounded as
\begin{align}\label{eq:varianceCSBound}
    Var(\hat{O}) &\leq \mathbb{E}_q \left[ \sum_{b\in\{0,1\}^n} p\left( b|U\right)\; \hat{O}^2\right] \nonumber\\
    & = \mathbb{E}_q \Big[ \sum_{b\in\{0,1\}^n} p(b\,|\,U) \nonumber \\
    & \hspace{0.2in}\times \left(\Tr\left[\mathcal{M}_p^{-1}(O) w(U) U^\dagger \ketbra{b}U \right]\right)^2 \Big]\nonumber \\
    & = \mathbb{E}_p \left[  w(U) f(U,O) \right]
\end{align}
where 
\begin{align}
    f(U,O) &:= \sum_{ b\in\{0,1\}^n}p(b\,|\,U)\nonumber \\
    & \hspace{0.4in} \times\left(Tr\{\mathcal{M}_p^{-1}(O)  U^\dagger \ketbra{b}U \}\right)^2. 
\end{align}
This gives a net cost 
\begin{align}
    NC_q \geq \left( \mathbb{E}_p \left[\sqrt{c(U) f(U,O)} \right]\right)^2\,, \label{eq:NC_CS_q}
\end{align}
where the inequality follows from the Cauchy-Schwartz inequality.
To attain the performance given by the above lower bound, 
one needs
\begin{align}
q^*(U) \propto p(U)\sqrt\frac{f(U,O)}{c(U)}\,.
\end{align}
However, this choice depends on the observable. To avoid this dependency, a heuristic choice for a sub-optimal distribution is
\begin{align}
    q'(U) \propto \frac{p(U)}{\sqrt{c(U)}}.
\end{align}
We demonstrate through an example below 
that this choice can also yield an improvement 
in the Net-Cost performance of Classical Shadows.
\paragraph{Example: 2-qubit Clifford Group with Pauli Observable}

For any 2-qubit Clifford $U$, let $n_{\mathrm{cx}}(U)\in\{0,1,2,3\}$ denote the number of CNOT gates in an optimal decomposition with respect to the number of CNOT gates. 
We assume an affine cost model with the per-shot cost as
\begin{align}
c(U) \;=\; \beta \;+\; n_{\mathrm{cx}}(U),
\end{align}
where $\beta$ is a fixed overhead (capturing single-qubit gates, reset, and readout) measured in ``CNOT-equivalent'' units.
Over the full 2-qubit Clifford group ($|\mathcal{C}_2|=11520$), the known CNOT distribution \cite{CorcolesRB2012} as summarized in Table \ref{table:CNOT2qubitDistr}, yields
\[
\mathbb{E}_{p}[\,c(U)\,] \;=\; \beta \;+\; 1.5.
\]

\begin{table}
    \centering
\begin{tabular}{|c|c|c|c|}
\hline
$n_{\mathrm{cx}}$ & \# Cliffords & Probability & Cost $c(U)$ \\
\hline
0 & 576   & $0.05$ & $\beta$ \\
1 & 5184  & $ 0.45$ & $\beta + 1$ \\
2 & 5184  & $0.45$ & $\beta + 2$ \\
3 & 576   & $0.05$ & $\beta + 3$ \\
\hline
\multicolumn{3}{|c|}{Average} & $\beta + 1.50$\\
\hline
\end{tabular}
    \caption{CNot distribution and the associated costs for two qubit Clifford circuits. The size of $\mathcal{C}_2$ is $11520.$}
    \label{table:CNOT2qubitDistr}
\end{table}
Moreover, any two-qubit unitary in $\mathrm{SU}(4)$ (hence any two-qubit Clifford) can be
implemented with at most three CNOTs, and this bound is tight~\cite{VatanWilliams2004}.
By assuming that $O$ is a Pauli observable, we observe that $\mathcal{M}_p^{-1}(O) = (d+1)O$ for all non-identity $O$. This implies 
\begin{align}
  &\Tr[ \mathcal{M}_p^{-1}(O)U^\dagger \ketbra{b}U ] = (d+1)\bra{b}UOU^\dagger\ket{b} \nonumber \\
  & \hspace{0.5in} = \begin{cases}
      (d+1) & if \quad UOU^\dagger \in \mathcal{D}_2\\ 0 & otherwise,
  \end{cases}
\end{align}
where $\mathcal{D}_2$ is the set of Pauli operators that are diagonal in the computational basis, i.e., $\mathcal{D}_2:= \{ Z\otimes I, I\otimes Z, Z\otimes Z\}$. Hence, we have 
\begin{align}
    f(U,O) = \begin{cases}
        (d+1)^2 & if \quad UOU^\dagger \in \mathcal{D}_2\\ 0 & otherwise,
    \end{cases}
\end{align}
For any non-identity two-qubit Pauli observable $O$,  conjugation by a uniformly random Clifford $U$ maps $O$ to a uniformly random non-identity Pauli. 
In some instances, one may determine a priori which Clifford unitaries $U$ yield $f(U,O)=0$. In such cases, the sampling distribution could in principle be modified to exclude these unitaries, as they do not contribute to the estimator for the given observable. However, in the present setting we refrain from doing so, both to retain an $O$-agnostic protocol and because classical shadows are typically employed to estimate multiple observables simultaneously, for which the set of unitaries satisfying $f(U,O)\neq 0$ may differ from observable to observable.
Since the two-qubit Pauli group has $4^2 - 1 = 15$ non-identity 
elements and exactly $|\mathcal{D}_2| = 3$ of them act non-trivially on the same support as $O$,  the probability that $U O U^\dagger$ remains nonzero on that support is $ 3/15 = 1/5.$ Let $n_i$ denote the number of Cliffords with $f(U,O) \neq  0$ having $i$ CNOTs in their decomposition. Using numerical simulations, we obtain $n_0=n_3 = 64$ and $n_1=n_2=1088$, independent of the Pauli observable $O$.
To avoid the dependency of the sampling distribution $q(U)$ on the observable, we choose
\begin{align}
    q'(U) \propto \frac{p(U)}{\sqrt{c(U)}}.
\end{align}
The Net-Cost obtained is given by
\begin{align}
    &NC_{q'} (\widehat{\langle O\rangle}) &=  \mathbb{E}_p \left[  \sqrt{c(U)} \right]\mathbb{E}_p \left[  \sqrt{c(U)} f(U,O) \right] \nonumber \\
\end{align}
Using the above arguments, 
we can compute
\begin{align}
    & \mathbb{E}_p \left[  \sqrt{c(U)} \right] = \nonumber \\ & 0.05(\sqrt{\beta}+\sqrt{\beta+3}) + 0.45(\sqrt{\beta+1}+ \sqrt{\beta+2}) \label{eq:CSCliffCost}
\end{align}
and
\begin{align}
    &\mathbb{E}_p \left[  \sqrt{c(U)} f(U,O) \right]  \nonumber \\
    & = \sum_{U \in \mathcal{C}_2: f(U,O)\neq 0} \frac{1}{|\mathcal{C}_2|} \sqrt{\beta + n_{\mathrm{cx}}(U)}{(d+1)^2}\nonumber \\
    & = {(d+1)^2 \over 11520}\sum_{i}n_i\sqrt{\beta+i} \nonumber \\
    & = \frac{25}{11520}\bigg(64 (\sqrt{\beta} + \sqrt{\beta+3}) \nonumber \\
     &\hspace{20pt}+ 1088(\sqrt{\beta+1} + \sqrt{\beta+2})\bigg)\,. \label{eq:CSCliffVar}
\end{align}
We can use this result to compute the ratio
$NC_{q'}/NC_p$ for different values of $\beta$.
In particular, we get $0.9828$ for $\beta=1$, and $0.9291$ for $\beta=0$, giving a relative gain of $7.09\%$ for $\beta=0$,
empirically proving the effectiveness of the IS ansatz.

\subsection{Probabilistic Error Cancellation}

Probabilistic error cancellation (PEC) was introduced in \cite{temme2017error}
as a quasi-probabilistic method for reconstructing noise-free expectation values
using noisy quantum hardware.
Consider a calibrated set of implementable noisy operations
$\Omega = \{ \mathcal{N}_{\utheta} \}$,
that spans the space of completely positive trace-preserving (CPTP) maps,
where each operation $\mathcal{N}_{\utheta}$ is associated with an implementation cost $c(\utheta)$.
Any ideal CPTP map $\mathcal{E}$ can be represented as a signed linear combination of such noisy operations, namely
\begin{align}
\mathcal{E} = \gamma \sum_{\utheta} p(\utheta)\, \eta(\utheta)\, \mathcal{N}_{\utheta},
\end{align}
where $\eta(\utheta) \in \{\pm 1\}$, $p(\utheta)$ is a probability density function,
and $\gamma \ge 1$ is a normalization factor.
As a consequence, 
one can construct an unbiased estimator of $\Tr[O \,\mathcal{E}(\rho)]$,
by sampling $\utheta$ from $p(\utheta)$,
implementing $\mathcal{N}_{\utheta}$ on the quantum hardware accordingly,
measuring the observable $O$ and reweighting each outcome
by the factor $\gamma\,\eta(\utheta)$.
This procedure is reminiscent of IS with $p$ playing the role of $q$ and $\gamma \eta(\utheta)$ playing the role of weights. The crucial difference is that here the weights can take negative values.
The variance of this estimator scales as $\mathcal{O}(\gamma^2)$,
implying that achieving an additive accuracy $\delta$ requires
$M = \mathcal{O}(\gamma^2 / \delta^2)$ circuit executions.
More generally, one may sample operations according to an alternative distribution $q(\utheta)$
and reweight each outcome by an additional factor
$w(\utheta) = p(\utheta) / q(\utheta)$.
This procedure remains unbiased and incurs an expected cost
$\mathbb{E}_q[c(\utheta)]$,
with variance bounded as
\begin{align}
\mathrm{Var} \le \gamma^2 \mathbb{E}_p[w(\utheta)] .
\end{align}
The sampling distribution that minimizes the overall net-cost is characterized by
Theorem~\ref{thm:optSampling} and reads
\begin{align}
q^*(\utheta) \propto \frac{p(\utheta)}{\sqrt{c(\utheta)}}
\end{align}
in absence of error-detection schemes.
This choice yields the net-cost improvement quantified in Eq.~\eqref{eq:netCostRatio}.


\subsubsection{Example: PEC for Depolarizing noise}
\label{subsec:pec_single_unitary_cost_is}
We here illustrate the advantages of optimal IS for probabilistic error cancellation (PEC) in the recovery of quantum channels affected by localized depolarizing noise.
To this end, we consider a collection of $k_g$-local unitary channels $\{\mathcal{U}_g\}_{g=1}^m$, whose action on a subset of $k_g$ qubits is given by
\begin{align}
\mathcal{U}_g(\rho) = U_g \rho U_g^\dagger.
\end{align}
The objective is to construct a quantum circuit by sequentially concatenating these unitaries and measure the observable 
$O$ at the end.
We assume that the available hardware instead realizes the noisy operations
\begin{align}
\mathcal{D}_{g,\varepsilon} \circ \mathcal{U}_g,
\end{align}
where $\mathcal{D}_{g,\varepsilon}$ is the local $k_g$-qubit depolarizing channel,
defined as
\begin{align}
\mathcal{D}_{g,\varepsilon}(\rho_g)
=
(1-\varepsilon)\rho_g
+
\varepsilon\,\frac{\mathrm{Tr}_g[\rho_g]}{2^{k_g}} I ,
\end{align}
where the trace is restricted to the $k_g$ qubits on which $\mathcal{U}_g$ acts non-trivially,
$\rho_g$ is the reduced density operator,
and the parameter $\varepsilon$ has been assumed to be independent from $g$.
To mitigate the effect of noise, we seek to invert the depolarizing channel.
This channel is diagonal in the Pauli operator basis: it leaves the identity invariant
and rescales any traceless Pauli operator by a factor $(1-\varepsilon)$.
Denoting the identity by $\sigma^0$ and the Pauli matrices $X$, $Y$, and $Z$
by $\sigma^1$, $\sigma^2$, and $\sigma^3$, respectively,
a generic Pauli $k_g$-string is written as
\begin{align}
P(\utheta_{g}) = \bigotimes_{i = 1}^{k_g} \sigma^{\theta_{g,i}},
\end{align}
where $\utheta_{g}$ determines the choice of Pauli string after each Unitary $U_g$, and $\theta_{g,i} \in \{0,1,2,3\}$. 
The action of the inverse depolarizing channel on this basis is
\begin{align}
\mathcal{D}_{g,\varepsilon}^{-1}(P(\utheta_{g})) =
\begin{cases}
P(\utheta_{g}) & \utheta_g = \underline{0}, \\
\frac{1}{1-\varepsilon} P(\utheta_{g}) & \text{otherwise}.
\end{cases}
\end{align}
The inverse map can be expressed as a linear combination of unitary conjugation channels
\begin{align}
\mathcal{P}_{\utheta_{g}}(\rho) := P(\utheta_{\,g})\,\rho\,P(\utheta_{\,g})
\end{align}
according to
\begin{align}
\mathcal{D}_{g,\varepsilon}^{-1}
&=
\sum_{\utheta_g} p_g(\utheta_g)\,\gamma_{g}\eta_g(\utheta_g)\mathcal{P}_{\utheta_g}\\
&=
\mathbb{E}_{p_g}\left[\gamma_{g}\eta_g(\utheta_g)\mathcal{P}_{\utheta_g}\right]\,,
\end{align}
where we introduced the probability density function
\begin{align}\label{eq:PEC_pdf}
    p_g(\utheta_g) =
    \begin{cases}
        \frac{4^{k_g}-\varepsilon}{4^{k_g}+(4^{k_g}-2)\varepsilon} & \utheta_g = \underline{0}\\
        \frac{\varepsilon}{4^{k_g}+(4^{k_g}-2)\varepsilon} & \utheta_g \neq \underline{0}\,,
    \end{cases}
\end{align}
the sign function
\begin{align}
    \eta_g(\utheta_g) =
    \begin{cases}
        1 & \utheta_g = \underline{0}\\
        -1 & \utheta_g \neq \underline{0}\,,
    \end{cases}
\end{align}
and the multiplicative factor
\begin{align}
\label{eq:gamma_n_single_unitary}
\gamma_g
:=
\frac{4^{k_g} + (4^{k_g}-2)\varepsilon}{4^{k_g}(1-\varepsilon)}\,.
\end{align}
Since $\mathcal{D}_{g,\varepsilon}^{-1}$ is not a completely positive trace-preserving map,
it cannot be implemented as a physical quantum channel.
Nevertheless, its action can be statistically reconstructed using PEC~\cite{temme2017error}.
Specifically, one samples the parameter vector $\utheta_g$ according to the distribution $p_g(\utheta_g)$ and implements on the quantum hardware the composition $\mathcal{P}_{\utheta_g}
\circ
\mathcal{U}_g$. The hardware adds the noise resulting in 
the channel
\begin{align}\label{eq:PEC_noisy_channel}
\mathcal{N}_{g}(\utheta_g)
=
\mathcal{D}_{g,\varepsilon}
\circ
\mathcal{P}_{\utheta_g}
\circ
\mathcal{U}_g.
\end{align}
After measuring the observable at the end of the circuit, the measurement outcome is reweighted by the factor $\gamma_g\,\eta_g(\utheta_g)$.
This procedure recovers,
in expectation,
the ideal channel, indeed
\begin{align}
\mathcal{U}_g
=
\gamma_g
\mathbb{E}_p
\left[\eta_g(\utheta_g)\,
(\mathcal{D}_{g,\varepsilon} \circ \mathcal{P}_{\utheta_g} \circ \mathcal{U}_g)
\right].
\end{align}
In what follows, 
we assume that (a) the same depolarizing parameter $\varepsilon$ for all $U_g$ gates (for notational simplicity),
(b) noise is Markovian and independent across gates, and
(c) the Pauli-twirl basis operations (Pauli conjugations) are easily implementable on the hardware.
For each gate $g\in[1,s]$ acting on $k_g$ qubits, the noisy quantum channel of Eq.~\eqref{eq:PEC_noisy_channel} is implemented independently.
Operationally, this corresponds to inserting a random Pauli $k_g$-string after each $\mathcal{U}_g$
according to the probability distribution $p_g$ prescribed in Eq.~\eqref{eq:PEC_pdf}
and multiplying each circuit outcome by a factor
\begin{align}
   \prod_{g = 1}^s \gamma_{g} \,\eta_{g}(\utheta_g) \,.
\end{align}
For simplicity, we denote the full \emph{parameter set} as $\Theta=(\utheta_1,...\,\utheta_m)$.
$\Theta$ determines both the choice of Pauli strings inserted after each gate
and the multiplicative factor of the circuit outcome.
Its probability distribution is a product distribution, 
simply given by
\begin{equation}
p(\Theta)
=\prod_{g=1}^s p_{g}(\utheta_g),
\label{eq:product_sampling}
\end{equation}
To switch to an IS approach,
it's enough to replace $p(\Theta)$ with an arbitrary $q(\Theta)$ probability distribution 
and to multiply each circuit outcome by an additional factor $w(\Theta) = p(\Theta)/q(\Theta)$.
This can be practically relevant whenever the circuits to be implemented 
have a cost $c(\Theta)$ that is highly dependent on the sampled vector $\Theta$.
In particular,
the net-cost optimal IS distribution over Pauli patterns is
\begin{equation}
q^\star(\Theta) \propto \frac{p(\Theta)}{\sqrt{c(\Theta)}}\,,
\label{eq:qstar_pattern}
\end{equation}
which allows to achieve a net-cost ratio
\begin{equation}\label{eq:PEC_net_cost_ratio}
\frac{\mathrm{NC}_{q^*}}{\mathrm{NC}_p}
=
\frac{\big(\mathbb{E}_{p}[c^{1/2}(\Theta)]\big)^2}{\mathbb{E}_{p}[c(\Theta)]}.
\end{equation}

In what follows, we present numerical results for a specific setup.
We consider a quantum circuit with an even number of qubits $n$ composed of $s$ layers. 
Each layer consists of: (i) $n$ single-qubit gates, and (ii) $n/2$ two-qubit gates acting on arbitrary disjoint pairs. 
We assume a \emph{gate-local depolarizing noise model}, 
in which each implemented gate $U_g$ acting on $k_g$ qubits is followed by a $k_g$-qubit depolarizing channel.
Applying PEC under this model results in a circuit structure as illustrated in Figure~\ref{fig:layer_structure}. 
Specifically, each single-qubit gate is followed by a Pauli operator ($k_g = 1$) sampled from 
\begin{align}
    p_g(\theta_g) = \begin{cases}
        \frac{4-\varepsilon}{4+2\varepsilon} & \theta_g = 0\\
        \frac{\varepsilon}{4+2\varepsilon} & \theta_g \neq 0\,,\\
    \end{cases}
\end{align} 
and each two-qubit gate is followed by a Pauli string of length $k_g = 2$ sampled from  
\begin{align}
    p_g(\utheta_g) = \begin{cases}
        \frac{16-\varepsilon}{16+14\varepsilon} & \utheta_g = \underline{0}\\
        \frac{\varepsilon}{16+14\varepsilon} & \utheta_g \neq \underline{0}\,.\\
    \end{cases}
\end{align} 
\begin{figure}[t]
    \centering
    \includegraphics[width=0.98\linewidth]{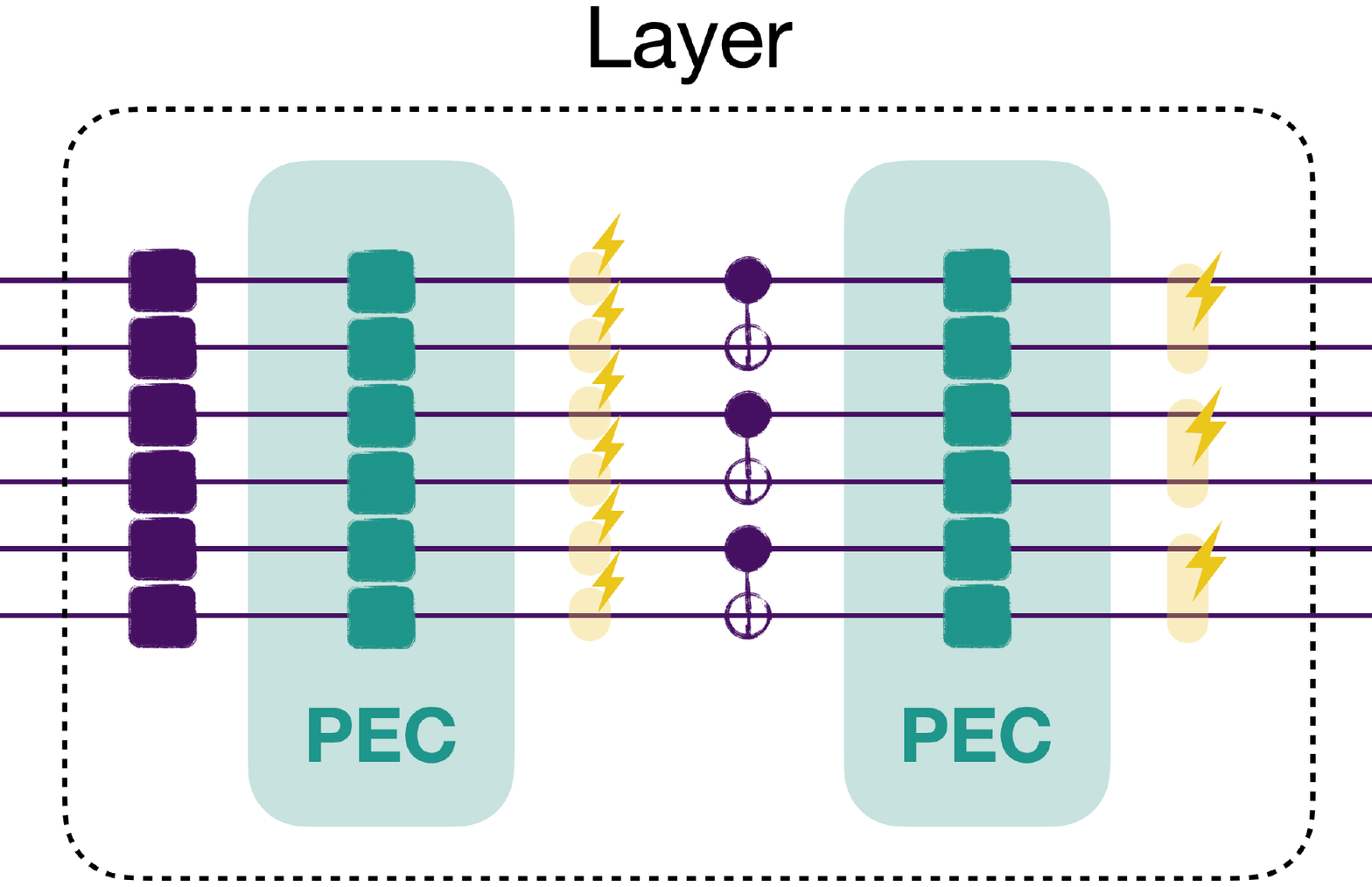}
    \caption{Structure of a single circuit layer. 
        The first stage consists of single-qubit gates applied to each qubit, 
        followed by 
        the randomized Pauli operations required to statistically cancel the noise
        and the localized single-qubit depolarizing channels. 
        Similarly, the second stage analogously applies two-qubit gates, after which
        randomized Pauli operations are performed
        and the $2$-local depolarizing channels act.}
    \label{fig:layer_structure}
\end{figure}
\begin{figure}[t]
    \centering
    \includegraphics[width=0.98\linewidth]{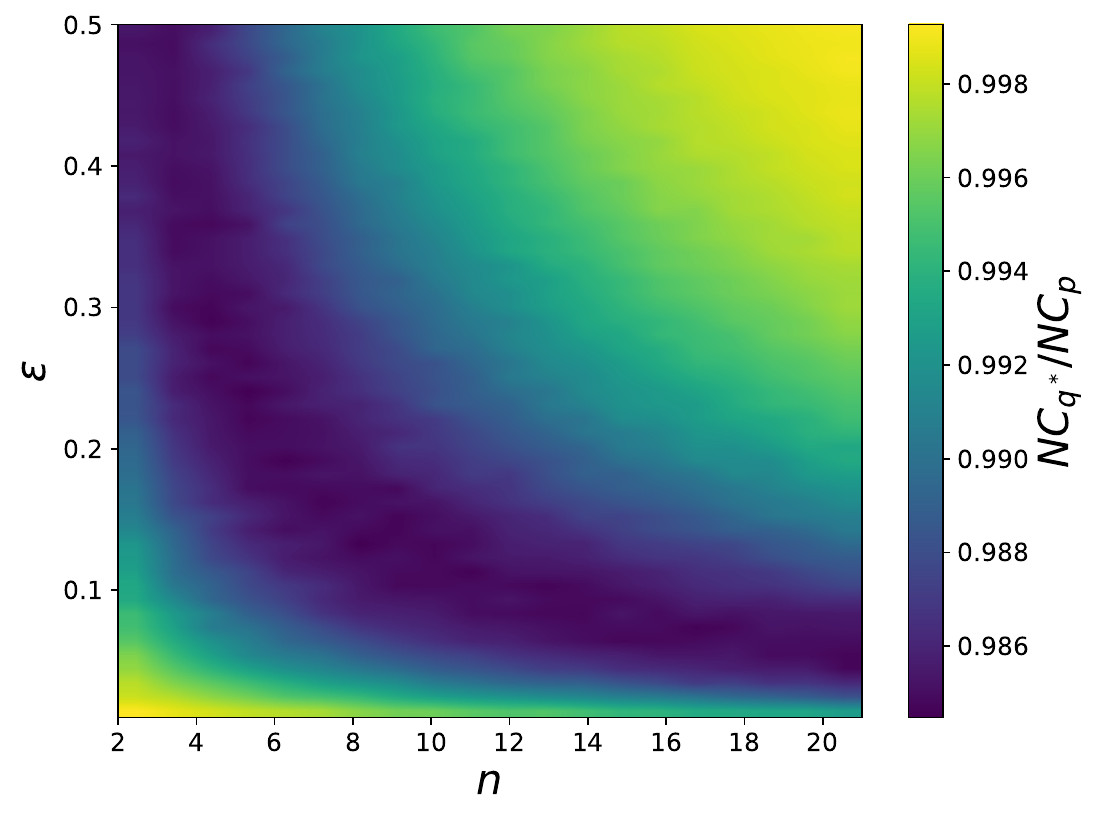}
    \caption{Net cost ratio $NC_{q^*}/NC_p$ for a single layer with the structure in Figure~\ref{fig:layer_structure}, as a function of the number of qubits $n$ and the depolarizing-noise parameter $\varepsilon$.}
    \label{fig:PEC_net_cost_ratio}
\end{figure}
Since all PEC Pauli operations are applied in parallel, 
both the single-qubit block and the two-qubit block are effectively followed by a random Pauli strings acting on all $n$ qubits, 
but according to different probability distributions.
For this example, 
we take the circuit depth as the cost metric and assume that single-qubit $I$ and $Z$ operations are “classical” 
(i.e., implemented via software frame updates or virtual-$Z$ gates), 
whereas $X$ and $Y$ require physical control pulses.
Consequently, 
each PEC Pauli string after a single/two qubit block contributes a cost of $0$ if it contains only $\sigma^0$ and $\sigma^3$, or $1$ otherwise.
As a result, 
the total cost of a layer with this structure, 
comprising both single- and two-qubit blocks, 
can take the values $2$, $3$, or $4$,
depending on the sampled Pauli strings. 
In Figure~\ref{fig:PEC_net_cost_ratio}, we report numerical results for the net cost ratio of Eq.~\eqref{eq:PEC_net_cost_ratio} 
for a single layer, 
as a function of both the noise parameter $\varepsilon$ and the number of qubits $n$. 
Remarkably, the minimal ratios,
i.e. the maximum advantages, are obtained for non-trivial combinations of $\varepsilon$ and $n$. 
Although this behavior may seem unexpected, 
it has a simple intuitive explanation. 
For a fixed number of qubits, 
in the limit $\varepsilon \to 0$, 
PEC is effectively unnecessary, 
and most sampled Pauli strings are identities. 
In this regime, cost fluctuations across different circuit runs vanish, 
and from Eq.~\eqref{eq:netCostRatio} we know that the net cost ratio approaches one. 
The same behavior occurs in the opposite regime of large $\varepsilon$. 
Here, the sampled Pauli strings are almost never identities, 
and the probability that a sampled string contains only $I$ or $Z$ is exponentially small in $n$. 
Consequently, almost all runs incur the maximum layer cost of $4$, 
with small statistical fluctuations, 
and the net cost ratio again approaches one. 
The optimal reduction in net cost therefore arises at intermediate noise levels and qubit numbers, 
where sampling fluctuations are significant and non-trivial low-cost Pauli strings still occur.
The generalization to $s$ layers is straightforward,
with the corresponding asymptotic scaling characterized by Theorem~\ref{thm:TheoremIIDLInfty}.

\section{Conclusions}
In this paper, we introduced a general framework that leverages IS as a classical optimization layer for randomized quantum algorithms, 
enabling systematic reductions in computational resource requirements without altering the underlying quantum procedures. 
Central to this contribution is an analytical expression for the optimal sampling distribution that minimizes the net cost associated with estimating a target quantity, 
thereby elevating IS from a heuristic variance-reduction technique,
to a principled resource-allocation mechanism for quantum computations.

We analytically established that the IS transformation preserves the effects of algorithmic imperfections and circuit noise.
As a consequence, 
while the framework can substantially reduce expected resource consumption through optimal sampling, 
it cannot be used to suppress or modify the effects of imperfect quantum operations, 

A defining feature of the framework is its flexibility with respect to cost modeling. 
The formalism accommodates arbitrary user-defined cost functions, 
allowing practitioners to tailor optimization to the dominant resource constraints of a given implementation. 
For near-term devices, this may correspond to minimizing two-qubit gate counts, 
whereas in fault-tolerant settings the relevant metric may instead reflect non-Clifford gate usage or other expensive logical operations. More broadly, the same methodology can be applied to optimize heterogeneous objectives such as wall-clock runtime or energy consumption, reinforcing the interpretation of IS as a general-purpose classical wrapper for resource-aware execution of randomized quantum protocols.

The framework is likewise observable-independent at the formulation level, yet it naturally admits observable-aware extensions. In particular, incorporating prior knowledge of a target observable into the optimization procedure could enable joint cost--variance tradeoffs, simultaneously reducing resource expenditure and statistical uncertainty. Developing such observable-informed strategies constitutes a promising direction for future work.

The magnitude of achievable cost reduction is governed by the statistical structure of per-circuit resource fluctuations. 
Substantial gains arise when the distribution of circuit costs is broad or when correlations between randomized steps preserve variability across runs. 
By contrast, in protocols composed of many independent randomized steps, where the per-run cost is additive in the cost of such steps, the net cost concentrates around its mean via the central limit theorem, leading to a progressive suppression of IS advantages that scales asymptotically as $1/s$ with the number of steps $s$.
This behavior highlights a fundamental limitation of the current formulation. 
Notably, the situation changes qualitatively in the presence of error-detection mechanisms: 
conditioning on successful runs effectively induces heavy-tailed cost statistics, 
under which the resource reduction enabled by IS can increase exponentially with $s$.
This observations suggests a strategy to circumvent the aforementioned limitation. 
For algorithms composed of multiple steps, one may consider randomized implementations that introduce correlations between choices across steps to avoid concentration effects due to the Central Limit Theorem. This approach could allow exploration of scenarios where the average cost in Eq.~\eqref{eq:netCostRatio} remains comparable to existing methods, while the average of the square root of the cost is reduced. Developing randomized algorithms with this type of structured variability therefore represents a promising direction for future work.

In Sec.~\ref{sec:applications}, we illustrated the versatility 
of our framework with several applications,
namely Hamiltonian simulation via randomized product formulas, 
implementation of dephasing channels via randomized time evolutions, mixed-state simulation, composite observable estimation, 
classical shadow tomography, and probabilistic error cancellation for error mitigation. 
While our examples employed simplified cost models, the framework can accommodate more realistic models and different optimization objectives, 
demonstrating its broad applicability across randomized quantum algorithms.
Since our framework provides a systematic procedure 
to reduce the computational resource requirement of any randomized protocol through optimal IS, 
we suggest that performance comparisons between deterministic and randomized algorithmic formulations 
should account for the use of IS. 
Therefore, 
a meaningful assessment requires contrasting the original algorithm 
with its randomized counterpart augmented by IS, 
rather than with the bare randomized implementation.

\section{Acknowledgments} 
This work was primarily led and supported by the U.S. Department of Energy, Office of Science, National Quantum Information Science Research Centers, Quantum Science Center (QSC). Y.S. was supported by QSC. All authors contributed to the developed of the theoretical framework and its applications and wrote the manuscript. T.A.A. acknowledges support by the Laboratory Directed Research and Development (LDRD) program of Los Alamos National Laboratory (LANL) under project number 20260043DR and the Information Science and Technology Institute (ISTI) program under ISTI Rapid Response. D.C. was a participant in the 2025 Quantum Computing Summer School at LANL, sponsored by the LANL Information Science \& Technology Institute. T.A.A and D.C. performed the numerical simulations. 
The authors warmly acknowledge Rolando Somma, Luca Spagnoli and Marco Cerezo for stimulating scientific discussions.

\bibliographystyle{IEEEtran}   
\bibliography{refs}

@article{CorcolesRB2012,
  author    = {C{\'o}rcoles, Antonio D. and Gambetta, Jay M. and Chow, Jerry M. 
               and Smolin, John A. and Ware, Marcus and Strand, Jared D. 
               and Plourde, Britton L. T. and Steffen, Matthias},
  title     = {Process verification of two-qubit quantum gates by randomized benchmarking},
  journal   = {arXiv preprint},
  eprint    = {1210.7011},
  year      = {2012},
  archivePrefix = {arXiv},
  primaryClass = {quant-ph}
}

@article{VatanWilliams2004,
  author    = {Vatan, Farrokh and Williams, Colin},
  title     = {Optimal Quantum Circuits for General Two-Qubit Gates},
  journal   = {Physical Review A},
  volume    = {69},
  number    = {3},
  pages     = {032315},
  year      = {2004},
  doi       = {10.1103/PhysRevA.69.032315},
  eprint    = {quant-ph/0308006},
  archivePrefix = {arXiv},
  primaryClass = {quant-ph}
}

@article{kiss2023importance,
  title={Importance sampling for stochastic quantum simulations},
  author={Kiss, Oriel and Grossi, Michele and Roggero, Alessandro},
  journal={Quantum},
  volume={7},
  pages={977},
  year={2023},
  publisher={Verein zur F{\"o}rderung des Open Access Publizierens in den Quantenwissenschaften}
}

@article{Childs2019,
  author    = {Andrew M. Childs and Aaron Ostrander and Yuan Su},
  title     = {Faster quantum simulation by randomization},
  journal   = {Quantum},
  volume    = {3},
  pages     = {182},
  year      = {2019},
  doi       = {10.22331/Q-2019-09-02-182},
  url       = {https://quantum-journal.org/papers/q-2019-09-02-182/}
}

@article{Campbell2019,
  author    = {Earl Campbell},
  title     = {Random Compiler for Fast Hamiltonian Simulation},
  journal   = {Phys. Rev. Lett.},
  volume    = {123},
  pages     = {070503},
  year      = {2019},
  doi       = {10.1103/PhysRevLett.123.070503}
}

@article{Wallman2016,
  author    = {Joel J. Wallman and Joseph Emerson},
  title     = {Noise tailoring for scalable quantum computation via randomized compiling},
  journal   = {Phys. Rev. A},
  volume    = {94},
  pages     = {052325},
  year      = {2016},
  doi       = {10.1103/PhysRevA.94.052325}
}

@article{FlammiaLiu2011,
  author    = {Steven T. Flammia and Yi{-}Kai Liu},
  title     = {Direct Fidelity Estimation from Few Pauli Measurements},
  journal   = {Phys. Rev. Lett.},
  volume    = {106},
  pages     = {230501},
  year      = {2011},
  doi       = {10.1103/PhysRevLett.106.230501}
}

@article{HuangKuengPreskill2020,
  author    = {Hsin{-}Yuan Huang and Richard Kueng and John Preskill},
  title     = {Predicting many properties of a quantum system from very few measurements},
  journal   = {Nat. Phys.},
  volume    = {16},
  pages     = {1050--1057},
  year      = {2020},
  doi       = {10.1038/s41567-020-0932-7}
}

@book{OwenMCBook,
  author    = {Art B. Owen},
  title     = {Monte Carlo Theory, Methods and Examples},
  year      = {2013},
  note      = {Online book},
  url       = {https://artowen.su.domains/mc/}
}

@book{LiuMCBook,
  author    = {Jun S. Liu},
  title     = {Monte Carlo Strategies in Scientific Computing},
  publisher = {Springer},
  year      = {2001},
  doi       = {10.1007/978-0-387-76371-2}
}

@article{TemmeBravyiGambetta2017,
  author    = {Kristan Temme and Sergey Bravyi and Jay M. Gambetta},
  title     = {Error Mitigation for Short-Depth Quantum Circuits},
  journal   = {Phys. Rev. Lett.},
  volume    = {119},
  pages     = {180509},
  year      = {2017},
  doi       = {10.1103/PhysRevLett.119.180509}
}

@article{EndoBenjaminLi2018,
  author    = {Suguru Endo and Simon C. Benjamin and Ying Li},
  title     = {Practical Quantum Error Mitigation for Near-Future Applications},
  journal   = {Phys. Rev. X},
  volume    = {8},
  number    = {3},
  pages     = {031027},
  year      = {2018},
  doi       = {10.1103/PhysRevX.8.031027}
}

@article{TakagiLimitsQEM2022,
  author    = {Ryuji Takagi and Suguru Endo and Shintaro Minagawa and Mile Gu},
  title     = {Fundamental limits of quantum error mitigation},
  journal   = {npj Quantum Information},
  volume    = {8},
  pages     = {114},
  year      = {2022},
  doi       = {10.1038/s41534-022-00618-z}
}

@article{RMPreviewQEM2023,
  author    = {Zhenyu Cai and Ryan Babbush and Simon C. Benjamin and Suguru Endo and William J. Huggins and Ying Li and Jarrod R. McClean and Thomas E. O'Brien},
  title     = {Quantum error mitigation},
  journal   = {Rev. Mod. Phys.},
  volume    = {95},
  number    = {4},
  pages     = {045005},
  year      = {2023},
  doi       = {10.1103/RevModPhys.95.045005}
}

@article{BoixoKnillSomma2009,
  author    = {Sergio Boixo and Emanuel Knill and Rolando D. Somma},
  title     = {Eigenpath traversal by phase randomization},
  journal   = {arXiv:0903.1652},
  year      = {2009},
  url       = {https://arxiv.org/abs/0903.1652}
}

@inproceedings{CunninghamRoland2024,
  author    = {Joseph Cunningham and J{\'e}r{\'e}mie Roland},
  title     = {Eigenpath Traversal by Poisson-Distributed Phase Randomisation},
  booktitle = {19th Conf.\ on the Theory of Quantum Computation, Communication and Cryptography (TQC 2024)},
  series    = {LIPIcs},
  volume    = {310},
  pages     = {7:1--7:20},
  year      = {2024},
  publisher = {Schloss Dagstuhl},
  doi       = {10.4230/LIPIcs.TQC.2024.7}
}

@article{temme2017error,
  title     = {Error Mitigation for Short-Depth Quantum Circuits},
  author    = {Temme, Kristan and Bravyi, Sergey and Gambetta, Jay M.},
  journal   = {Physical Review Letters},
  volume    = {119},
  number    = {18},
  pages     = {180509},
  year      = {2017},
  doi       = {10.1103/PhysRevLett.119.180509}
}

@article{hilder2022fault,
  title={Fault-tolerant parity readout on a shuttling-based trapped-ion quantum computer},
  author={Hilder, Janine and Pijn, Daniel and Onishchenko, Oleksiy and Stahl, Alexander and Orth, Maximilian and Lekitsch, Bj{\"o}rn and Rodriguez-Blanco, Andrea and M{\"u}ller, Markus and Schmidt-Kaler, Ferdinand and Poschinger, UG},
  journal={Physical Review X},
  volume={12},
  number={1},
  pages={011032},
  year={2022},
  publisher={APS}
}

@article{self2024protecting,
  title={Protecting expressive circuits with a quantum error detection code},
  author={Self, Chris N and Benedetti, Marcello and Amaro, David},
  journal={Nature Physics},
  volume={20},
  number={2},
  pages={219--224},
  year={2024},
  publisher={Nature Publishing Group UK London}
}

@article{campbell2018random,
  title={A random compiler for fast Hamiltonian simulation},
  author={Campbell, Earl},
  journal={arXiv preprint arXiv:1811.08017},
  year={2018}
}

@book{mitzenmacher2017probability,
  title={Probability and computing: Randomization and probabilistic techniques in algorithms and data analysis},
  author={Mitzenmacher, Michael and Upfal, Eli},
  year={2017},
  publisher={Cambridge university press}
}

@article{sanders2020compilation,
  title = {Compilation of Fault-Tolerant Quantum Heuristics for Combinatorial Optimization},
  author = {Sanders, Yuval R. and Berry, Dominic W. and Costa, Pedro C.S. and Tessler, Louis W. and Wiebe, Nathan and Gidney, Craig and Neven, Hartmut and Babbush, Ryan},
  journal = {PRX Quantum},
  volume = {1},
  issue = {2},
  pages = {020312},
  numpages = {70},
  year = {2020},
  month = {Nov},
  publisher = {American Physical Society},
  doi = {10.1103/PRXQuantum.1.020312},
  url = {https://link.aps.org/doi/10.1103/PRXQuantum.1.020312}
}

@article{linke2018measuring,
  title={Measuring the R{\'e}nyi entropy of a two-site Fermi-Hubbard model on a trapped ion quantum computer},
  author={Linke, Norbert M and Johri, Sonika and Figgatt, Caroline and Landsman, Kevin A and Matsuura, Anne Y and Monroe, Christopher},
  journal={Physical Review A},
  volume={98},
  number={5},
  pages={052334},
  year={2018},
  publisher={APS}
}

@article{bonet2018low,
  title={Low-cost error mitigation by symmetry verification},
  author={Bonet-Monroig, Xavi and Sagastizabal, Ramiro and Singh, M and O'Brien, TE},
  journal={Physical Review A},
  volume={98},
  number={6},
  pages={062339},
  year={2018},
  publisher={APS}
}

@article{cai2021quantum,
  title={Quantum error mitigation using symmetry expansion},
  author={Cai, Zhenyu},
  journal={Quantum},
  volume={5},
  pages={548},
  year={2021},
  publisher={Verein zur F{\"o}rderung des Open Access Publizierens in den Quantenwissenschaften}
}

@article{mcardle2019error,
  title={Error-mitigated digital quantum simulation},
  author={McArdle, Sam and Yuan, Xiao and Benjamin, Simon},
  journal={Physical review letters},
  volume={122},
  number={18},
  pages={180501},
  year={2019},
  publisher={APS}
}

@article{arrasmith2020operator,
  title={Operator sampling for shot-frugal optimization in variational algorithms},
  author={Arrasmith, Andrew and Cincio, Lukasz and Somma, Rolando D and Coles, Patrick J},
  journal={arXiv preprint arXiv:2004.06252},
  year={2020}
}

@software{pennylane,
  author = {Bergholm, Ville and Izaac, Josh and Schuld, Maria and Killoran, Nathan and Wierichs, David and Gogolin, Christian and Blank, Christian and Wecker, Dave and Troyer, Matthias and Petruccione, Francesco and Aspuru-Guzik, Alan and Wang, Zixuan and Cincio, Lukasz and Meyer, Jens},
  title = {PennyLane: Automatic differentiation of hybrid quantum-classical computations},
  year = {2018},
  url = {https://pennylane.ai/},
  doi = {10.5281/zenodo.2562110}
}

@article{ransford2025helios,
  title={Helios: A 98-qubit trapped-ion quantum computer},
  author={Ransford, Anthony and Allman, MS and Arkinstall, Jake and Campora III, JP and Cooper, Samuel F and Delaney, Robert D and Dreiling, Joan M and Estey, Brian and Figgatt, Caroline and Hall, Alex and others},
  journal={arXiv preprint arXiv:2511.05465},
  year={2025}
}

@article{jennings2025randomized,
  title = {Randomized Adiabatic Quantum Linear Solver Algorithm with Optimal Complexity Scaling and Detailed Running Costs},
  author = {Jennings, David and Lostaglio, Matteo and Pallister, Sam and Sornborger, Andrew T. and Suba\ifmmode \mbox{\c{s}}\else \c{s}\fi{}\ifmmode \imath \else \i \fi{}, Yi\ifmmode \breve{g}\else \u{g}\fi{}it},
  journal = {PRX Quantum},
  volume = {6},
  issue = {4},
  pages = {040373},
  numpages = {23},
  year = {2025},
  month = {Dec},
  publisher = {American Physical Society},
  doi = {10.1103/1xkb-22cc},
  url = {https://link.aps.org/doi/10.1103/1xkb-22cc}
}

\clearpage
\onecolumngrid
\appendix

\section{Proof of Theorem \ref{thm:algorithms performancesZeroFillDiscardL}}
\label{appx:proofZeroFillDiscardL}
We begin with a proof for the ZeroFill$(L)$ algorithm.
For a fixed value of $\utheta$, each circuit attempt succeeds with probability $f(\utheta)$. Since the circuit is repeated at most $L$ times, the probability that 
all $L$ runs fail is $
\bigl(1-f(\utheta)\bigr)^L.
$ Therefore, the probability that $\utheta$ yields at least one successful 
measurement is
\[
k(\utheta;L) 
\coloneq 1 - \bigl(1-f(\utheta)\bigr)^L.
\]
Assume that $p(\utheta)$ is absolutely continuous with respect to $q(\utheta)k(\utheta;L)$,
the estimator is unbiased for every observable $O$ if
\begin{align}
\mathbb{E}_{qr}\left[w_Z(\underline{\theta})k(\utheta;L)x\right] \nonumber = \mathbb{E}_{pr}\left[x\right].
\end{align}
This holds (a.e.) if and only if
\begin{align}\label{eq:unbiasedwDetectError}
    w_Z(\underline{\theta}) &= \frac{p(\utheta)}{q(\utheta)k(\utheta;L)} = \frac{p(\underline{\theta})}{q(\underline{\theta})\left[1- (1-f(\underline{\theta}))^L \right]}\,.
\end{align}
To characterize the average cost per run, we note the following.
Given a $\utheta$, the number of attempts required to record an outcome is $\min(J,L)$, where $J \sim \mathrm{Geom}(f(\underline{\theta}))$ 
counts the number of trials until the first success, and $L$ is the fixed positive integer limiting the maximum number of allowed trials for each $\underline\theta$.  
Using the tail-sum formula,
\begin{align}
    \mathbb{E}[\min(J,L)|\utheta] = \sum_{j=1}^{\infty} \Pr[\min(J,L) \ge j|\utheta].
\end{align}
For $1 \le j \le L$, the event $\min(J,L) \ge j$ is equivalent to 
$J \ge j$. Hence 
\begin{align}
    \Pr[\min(J,L) \ge j|\utheta] &= \Pr[J \ge j|\utheta]  = (1-f(\underline\theta))^{\,j-1}.
\end{align}
For $j > L$, we cannot have $\min(J,L) \ge j$, so the probability is zero. Thus, the expected number of circuit runs given a $\utheta$ is
\begin{align}
    \mathbb{E}[\min(J,L)|\utheta] &= \sum_{j=1}^L (1-f(\underline\theta))^{j-1} = \frac{1-(1-f(\underline\theta))^L}{f(\underline\theta)} =\frac{k(\utheta;L)}{f(\utheta)},\label{eq:expecAttempts}
\end{align}
and the expected cost is $c(\utheta)k(\utheta;L)\over f(\utheta)$.
Averaging over the sampling distribution $q(\utheta)$ gives the expected cost per run as
\begin{align}
& \mathbb{E}_{q}\!\left[\, c(\utheta)\,\frac{k(\utheta;L)}{f(\underline\theta))} \right]. 
\end{align}
For the Variance, we have
\begin{align}
    Var_{qr}^{\max} & \leq \max_{|x| \leq 1}\mathbb{E}_{qr}\big[ k(\utheta;L)  w_Z^2(\utheta)x^2 \big]  = \mathbb{E}_q\left[k(\utheta;L)w_Z^2(\utheta)\right].
\end{align}
Therefore,
\begin{align}
    NC_q &= \mathbb{E}_{q}\!\left[\, c(\utheta)\,\frac{k(\utheta;L)}{f(\underline\theta)} \right] \mathbb{E}_q\left[k(\utheta;L)w_Z^2(\utheta)\right].
\end{align}
Using the condition \eqref{eq:unbiasedwDetectError}, we perform the change of measure to obtain
\begin{align}
    NC_q & \;=\; \mathbb{E}_{p}\!\left[\frac{ c(\utheta)}{w_Z(\utheta)f(\underline\theta)} \right] \mathbb{E}_p\left[w_Z(\utheta)\right]  \;\geq\; \left(\mathbb{E}_p\left[ \sqrt{\frac{c(\utheta)}{f(\underline\theta)}}\right]\right)^2,
\end{align}
where the last inequality follows from the Cauchy-Schwarz inequality, with equality if and only if $w_Z(\utheta)\propto\sqrt{c(\utheta)/f(\utheta)}$ which implies:
\begin{align}
    q(\utheta) \propto \frac{p(\utheta)}{k(\utheta;L)} \sqrt{\frac{f(\utheta)}{c(\utheta)}} = \frac{p(\utheta)}{ 1 - \bigl(1-f(\utheta)\bigr)^L} \sqrt{\frac{f(\utheta)}{c(\utheta)}}.
\end{align}

We now provide a proof for the Discard$(L)$ algorithm. In this algorithm $\utheta$ is sampled from $q(\utheta)$ and if errors are detected in all $L$ runs of the corresponding quantum circuit, a new parameter $\utheta$ is sampled from $q(\utheta)$. This is in fact rejection sampling and results in $\utheta$ begin sampled from a new distribution given by:
\begin{align}
h_L(\underline{\theta})\coloneq\;\frac{q(\underline{\theta})\,k(\underline{\theta},L)}{\mathbb{E}_q[k(\utheta;L)]},
\end{align}
where $  \mathbb{E}_q[k(\utheta;L)] = \int q(\underline{\theta})\,k(\underline{\theta};L)\,d\underline{\theta}$ is the acceptance rate. 
Now we can think of the procedure as the standard IS algorithm with $q(\utheta)\rightarrow h_L(\utheta)$ and for an unbiased estimator for any observable $O$, we choose 
\begin{align}\label{eq:discard_unbiasedwDetectError}
    w_D(\underline{\theta}) = p(\utheta)/h_L(\utheta)\,.
\end{align}
Because of the rejection sampling step, we need to sample $\utheta$ from $q(\utheta)$ on average $1/\mathbb{E}_q[k(\utheta;L)]$ times in order to get a single sample from $h_L(\utheta)$.  
From \eqref{eq:expecAttempts}, for a given \(\underline{\theta}\) the expected number of circuit runs is $k(\utheta;L)/f(\utheta)$. 
Hence the expected cost per sample is
\begin{align}
     \frac{1}{\mathbb{E}_q[k(\utheta;L)]}\mathbb{E}_q\left[c(\underline{\theta})\,\frac{k(\underline{\theta};L)}{f(\underline{\theta})}\right] = \mathbb{E}_{h_L}\left[\frac{c(\underline{\theta})}{f(\underline{\theta})}\right].
\end{align}
As for the Variance, we have
\begin{align}
    Var_{qr}^{\max} 
    & \leq \mathbb{E}_{h_L}\left[w_D^2(\utheta)\right].
\end{align}
This gives the bound on the net cost as
\begin{align}
    NC_q & = \mathbb{E}_{h_L}\left[\frac{c(\utheta)}{f(\utheta)}\right]\mathbb{E}_{h_L}\left[w_D^2(\utheta)\right]  =\mathbb{E}_p\left[\frac{c(\utheta)}{w_D(\utheta)f(\utheta)}\right]\mathbb{E}_p\left[w_D(\utheta)\right] \geq \left(\mathbb{E}_p\left[\sqrt\frac{c(\utheta)}{f(\utheta)}\right]\right)^2,
\end{align}
where the last inequality follows from the Cauchy-Schwarz inequality, with equality if and only if $w(\utheta)\propto\sqrt{c(\utheta)/f(\utheta)}$ which implies:
\begin{align}
    q(\utheta) & \propto \frac{p(\utheta)}{k(\utheta;L)}\sqrt\frac{f(\utheta)}{c(\utheta)} = \frac{p(\utheta)}{1-(1-f(\utheta))^L}\sqrt\frac{f(\utheta)}{c(\utheta)}.
\end{align}
This completes the proof.

\section{Proof of Theorem \ref{thm:TheoremIIDLInfty}}
\label{proof:TheoremIIDLInfty}
In the limit of large $s$, the Central Limit Theorem applies,
implying that the distribution of $c(\underline{\theta})$ becomes close to a normal distribution $\mathcal{N}\!\big(c;\,s\mu_1, s\sigma^2_1\big)$,
where $\mu_1$ and $\sigma^2_1$ are the first two cumulants 
of each variable $c_1(\theta_i)$.
By the Berry--Esseen theorem, if the third absolute moment of $c_1(\theta_i)$ is finite, then the cumulative distribution function of the normalized sum differs from the Gaussian CDF $\Phi$ by at most a constant times $1/\sqrt{s}$. Concretely, there exists a finite constant $C$ (depending only on the third absolute moment) such that for all real $y$
\[
\Big|\mathbb{P}\Big(\frac{c-s\mu_1}{\sqrt{s}\,\sigma_1}\le y\Big)-\Phi(y)\Big|\le \frac{C}{\sqrt{s}}\,.
\]
To propagate this CDF bound to expectations of the (unbounded) test function $g(c)=c^\alpha e^{\lambda\alpha c}$ we fix a cutoff $M>0$ and write
\[
\mathbb{E}[g(c)] = \mathbb{E}[g(c)\mathbf{1}_{c\le M}] + \mathbb{E}[g(c)\mathbf{1}_{c>M}].
\]
We apply Berry--Esseen to the bounded function $g\mathbf{1}$ to obtain an additive $O(s^{-1/2})$ error for the truncated expectation; then choose $M$ large enough (using the assumed moment bound) so that the tail $\mathbb{E}[g(c)\mathbf{1}_{c>M}]$ is uniformly small. Letting $M\to\infty$ after the Berry--Esseen step implies that replacing the exact distribution of $c$ by the Gaussian $\mathcal{N}(c;\,s\mu_1,s\sigma_1^2)$ incurs an additive error of order $O(s^{-1/2})$ in $\mathbb{E}[g(c)]$. Thus the Gaussian approximation step contributes an $O(s^{-1/2})$ additive error in the final asymptotic evaluation, i.e.
\begin{align}
\mathbb{E}_p\left[\left(\frac{c(\utheta)}{f(\utheta)}\right)^\alpha\right]
    & = \frac{\displaystyle\int_0^\infty dc\; c^\alpha e^{\lambda \alpha c} \exp\!\big(-\tfrac{(c-s\mu_1)^2}{2s\sigma_1^2}\big)}
    {\displaystyle\int_0^\infty dc\; \exp\!\big(-\tfrac{(c-s\mu_1)^2}{2s\sigma_1^2}\big)} + \mathcal{O}\left(\frac{1}{\sqrt{s}}\right),
\end{align}
where the normal density is truncated at $c=0$ and the denominator enforces the normalization on $[0,\infty)$. Here we introduce the parameter $\alpha$: setting $\alpha=0.5$ evaluates the net cost under the sampling distribution $q(\utheta)$ (see \eqref{eq:NetCostZeroFill}), whereas $\alpha = 1$ evaluates the net cost under the target distribution $p(\utheta)$ (see \eqref{eq:netCostLinftyWithP}).
Since we are studying the asymptotic behavior for large $s$, we first control the denominator with a Mills-type expansion. Writing $c=s\mu_1+\sqrt{s}\,\sigma_1 x$ gives
\begin{align*}
   \int_0^\infty dc\;  \exp\!\Big(-\frac{(c-s\mu_1)^2}{2s\sigma_1^2}\Big)& = \sqrt{s}\,\sigma_1\int_{-\sqrt{s}\mu_1/\sigma_1}^\infty dx\; e^{-x^2/2}\\
 &= \sqrt{2\pi s \sigma_1^2}\; \Phi\!\Big(\frac{\sqrt{s}\,\mu_1}{\sigma_1}\Big)\\
 &= \sqrt{2\pi s \sigma_1^2}\,\Big[1+ \mathrm{o}\!\big(e^{-s\mu_1^2/(2\sigma_1^2)}\big)\Big]\,,
\end{align*}
Hence the denominator tends to $\sqrt{2\pi s\sigma_1^2}$ up to an exponentially small correction.
First, we focus on the numerator
and complete the square in the exponent:
\begin{align*}
    & \int_0^\infty dc\; c^\alpha e^{\lambda \alpha c} e^{-(c-s\mu_1)^2/(2s\sigma_1^2)} \nonumber \\
    &= \int_0^\infty dc\; c^\alpha \exp\!\Big(-\frac{1}{2s\sigma_1^2}\Big(c-\big[s\mu_1+s \lambda \alpha \sigma_1^2\big]\Big)^2\Big)\exp\Big(\frac{1}{2s\sigma_1^2}\Big(\big[s\mu_1+s \lambda \alpha \sigma_1^2\big]^2-s^2\mu_1^2\Big)\Big)\\[4pt]
    &= \exp\!\Big(s\alpha\big(\lambda\mu_1+\tfrac{\alpha}{2}\lambda^2\sigma_1^2\big)\Big) 
    \int_0^\infty \!\!dc c^\alpha \exp\!\Big(\!-\frac{1}{2s\sigma_1^2}\Big(c-\big[s\mu_1+s\lambda\alpha\sigma_1^2\big]\Big)^2\Big).
\end{align*}
We now put numerator and denominator together,
obtaining
\begin{align}
&\mathbb{E}_p\left[\left(\frac{c(\utheta)}{f(\utheta)}\right)^\alpha\right] = \exp\!\Big(s\alpha\big(\lambda\mu_1+\tfrac{\alpha}{2}\lambda^2\sigma_1^2\big)\Big)\,
      \Big[1+ \mathrm{o}\!\big(e^{-s\alpha\mu_1^2/(2\sigma_1^2)}\big)\Big]      \int_0^\infty dc\; c^\alpha\; \mathcal{N}\!\big(c;\,s\mu_1+s\lambda\alpha\sigma_1^2,\; s\sigma_1^2\big).
\end{align}
The remaining term to be evaluated is the integral
\begin{align}\label{eq: finite c alpha exp}
&\int_0^\infty dc\; c^\alpha\; \mathcal{N}\!\big(c;\,s\mu_1+s\lambda\alpha\sigma_1^2,\; s\sigma_1^2\big)
= \sigma_1^{\alpha}s^{\alpha/2}\hspace{-0.4in}\int\limits_{-\sqrt{s}[\mu_1+\lambda\alpha\sigma_1^2]/\sigma_1}^{\infty} \hspace{-0.4in} dx\;
\Big(\frac{\sqrt{s}\mu_1+\sqrt{s}\lambda\alpha\sigma_1^2}{\sigma_1}+x\Big)^\alpha\frac{e^{-x^2/2}}{\sqrt{2\pi}}.
\end{align}
The lower limit $-\sqrt{s}[\mu_1+\lambda \alpha \sigma_1^2]/\sigma_1$ tends to $-\infty$ as $s\to\infty$ when $\mu_1>0$. Extending the integral down to $-\infty$ introduces an error equal to the discarded left tail
\begin{align}
    &T(s)\coloneq \Big|\int\limits_{-\infty}^{-\sqrt{s}[\mu_1+\lambda \alpha \sigma_1^2]/\sigma_1} dx\;
\left(\frac{\sqrt{s}\mu_1+\sqrt{s}\lambda\alpha\sigma_1^2}{\sigma_1}+x\right)^\alpha \frac{e^{-x^2/2}}{\sqrt{2\pi}}\Big|.
\end{align}
We bound $T(s)$ by Cauchy--Schwarz as follows. We define
\begin{align}
    x_0\coloneq &\frac{\sqrt{s}[\mu_1+\lambda \alpha \sigma_1^2]}{\sigma_1}>0,
\qquad \phi(x)\coloneq\frac{e^{-x^2/2}}{\sqrt{2\pi}}, \qquad h(x)\coloneq\left[\frac{\sqrt{s}\mu_1+\sqrt{s}\lambda\alpha\sigma_1^2}{\sigma_1}+x\right]^\alpha.
\end{align}
Then by Cauchy--Schwarz,
\begin{align}
    &T(s)=\abs{\int_{-\infty}^{-x_0} h(x)\,\phi(x)\,dx} \le \Big(\int_{-\infty}^{-x_0} h(x)^2\phi(x)\,dx\Big)^{1/2}\;
      \Big(\int_{-\infty}^{x_0}\phi(x)\,dx\Big)^{1/2}.
\end{align}
We henceforth estimate each factor: (i) Tail probability factor (second factor). Using the standard bound for the Gaussian tail,
\begin{align}
    \int_{-\infty}^{-x_0}\phi(x)\,dx&=\int_{x_0}^{\infty}\phi(x)\,dx
\le \frac{\phi(x_0)}{x_0} = \frac{1}{x_0\sqrt{2\pi}}\,e^{-x_0^2/2}.
\end{align}
Hence its square root satisfies
\[
\Big(\int_{-\infty}^{-x_0}\phi(x)\,dx\Big)^{1/2}
\le \frac{1}{(2\pi)^{1/4}}\,x_0^{-1/2}\,e^{-x_0^2/4}.
\]

(ii) First factor. For $x\le -x_0$ we may bound $h(x)^2$ by a polynomial times $s^{\alpha}$. Indeed there exist constants (depending only on $\alpha,\mu_1,\lambda,\sigma_1$) such that
\[
h(x)^2 \le C\,s^{\alpha}\,(1+x^{2\alpha}).
\]
Hence
\begin{align}
    \int_{-\infty}^{-x_0} h(x)^2\phi(x)\,dx &\le C\,s^{\alpha}\int_{x_0}^{\infty}(1+x^{2\alpha})\phi(x)\,dx
 \leq C\,s^{\alpha}\int_{-\infty}^{\infty}(1+x^{2\alpha})\phi(x)\,dx.
\end{align}
For $\alpha$ integer or half integer, 
the r.h.s. integral is a sum of moments from a normal distributions,
hence it's finite.
Combining (i) and (ii) in the Cauchy--Schwarz inequality yields
\[
T(s) = O\!\Big( s^{(2\alpha-1)/2}\,\exp\!\Big(-\frac{s[\mu_1+\lambda\alpha\sigma_1^2]^2}{2\sigma_1^2}\Big)\Big).
\]
Thus the error from extending the lower limit in Eq.~\eqref{eq: finite c alpha exp} to $-\infty$ is exponentially small in $s$ and negligible compared to any algebraic powers of $s$.
Hence, the final step is to tackle the integral  
\[
\int_{-\infty}^{\infty} dx\;
\Big(\frac{\sqrt{s}\mu_1+\sqrt{s}\lambda\alpha\sigma_1^2}{\sigma_1}+x\Big)^\alpha\frac{e^{-x^2/2}}{\sqrt{2\pi}}.
\]
At this scope we apply Laplace's method by Taylor-expanding the $h(x)$ function around $x=0$, that is around the maximum point of the Gaussian. 
Its Taylor expansion up to second order about $x=0$ reads
\begin{align*}
    h(x) &=\Big(\frac{\sqrt{s}\mu_1+\sqrt{s}\lambda\alpha\sigma_1^2}{\sigma_1}\Big)^\alpha
     \Bigg[1 + \alpha\Big(\frac{\sqrt{s}\mu_1+\sqrt{s}\lambda\alpha\sigma_1^2}{\sigma_1}\Big)^{-1}x  + \frac{\alpha(\alpha-1)}{2}\Big(\frac{\sqrt{s}\mu_1+\sqrt{s}\lambda\alpha\sigma_1^2}{\sigma_1}\Big)^{-2}\!\!x^2\Bigg] + R(x),
\end{align*}
where $R(x)$ is the remainder.
If we neglect the contribution of $R(x)$,
the consequent error to the final result will be
\[
\int_{-\infty}^{\infty} dx\;
R(x)\frac{e^{-x^2/2}}{\sqrt{2\pi}}\,.
\]
We henceforth bound the integral of the remainder 
against the Gaussian density 
by splitting the integral into the two regions
\[
\text{(A)}\quad |x|\le s^{1/4},\qquad \text{(B)}\quad |x|> s^{1/4}.
\]
\paragraph{Region A: small $x$.}
For $|x|\le s^{1/4}$ we have
\[
\Big|\frac{x}{\sqrt{s}(\mu_1+\lambda\alpha\sigma_1^2)/\sigma_1}\Big|
= O(s^{-1/4}),
\]
so for sufficiently large $s$ the expansion parameter is uniformly small. By the Lagrange form of Taylor's remainder,
$\exists \xi(x) \in [0,x] \subseteq [0, s^{1/4}]$
\[
|R(x)|
= \abs{\frac{h^{(3)}\xi(x)}{3!}x^3} \leq C_\alpha \abs{\frac{x}{\sqrt{s}(\mu_1+\lambda\alpha\sigma_1^2)/\sigma_1}}^3\,,
\]
with a constant $C_\alpha$ depending only on $\alpha$ and the fixed combination $\mu_1+\lambda\alpha\sigma_1^2$. Therefore, integrating $|x|^3 e^{-x^2/2}$ for $|x|\le s^{1/4}$ yields
\[
\Big|\int_{|x|\le s^{1/4}} \frac{R(x)}{\sqrt{2\pi}} e^{-x^2/2} dx\Big|
= O\!\big(s^{-3/2}\big).
\]

\paragraph{Region B: large $x$.}
On $|x|>s^{1/4}$ we bound the contribution by Cauchy--Schwarz:
\begin{align}
    &\Big|\int_{|x|>s^{1/4}} \frac{R(x)}{\sqrt{2\pi}} e^{-x^2/2} dx\Big| \le \Big(\int_{|x|>s^{1/4}} \frac{R(x)^2}{\sqrt{2\pi}} e^{-x^2/2} dx\Big)^{1/2} 
      \Big(\int_{|x|>s^{1/4}} \frac{1}{\sqrt{2\pi}} e^{-x^2/2} dx\Big)^{1/2}.
\end{align}
The tail probability factor is $O\!\big(e^{-s^{1/2}/2}\big)$. 
On the other hand, since $R(x)$ is a difference of polynomials,
it is also a polynomial, 
and hence the first integral can be bounded with a finite value.
Hence the overall Region B contribution is $O\!\big(e^{-s^{1/2}/2}\big)$.

\paragraph{Total remainder.}
Combining Region A and Region B bounds we obtain the explicit remainder estimate for the Laplace step:
\[
\int_{-\infty}^{\infty} \frac{R(x)}{\sqrt{2\pi}} e^{-x^2/2} dx
= O\!\big(s^{-3/2}\big),
\]
since the algebraic term dominates the exponentially small term for large $s$.
Thus the integral over the whole real line 
satisfies the Laplace expansion
\begin{align}
    &\int_{-\infty}^{\infty} dx\; \frac{h(x)}{\sqrt{2\pi}}e^{-x^2/2}
= \Big(\frac{\sqrt{s}\mu_1+\sqrt{s}\lambda\alpha\sigma_1^2}{\sigma_1}\Big)^\alpha 
\Bigg[1+\frac{\alpha(\alpha-1)}{2}\Big(\frac{\mu_1+\lambda\alpha\sigma_1^2}{\sigma_1}\Big)^{-2}\frac{1}{s}
+ O\!\big(s^{-3/2}\big)\Bigg].
\end{align}
The final asymptotic expansion is
\begin{align}
\mathbb{E}_p\left[\left(\frac{c(\utheta)}{f(\utheta)}\right)^\alpha\right]\nonumber = \exp\!\Big(s\alpha\big(\lambda\mu_1+\tfrac{\alpha}{2}\lambda^2\sigma_1^2\big)\Big) s^{\alpha}\Big(\mu_1+\lambda\alpha\sigma_1^2\Big)^\alpha 
\Bigg[1+\frac{\alpha(\alpha-1)}{2}\Big(\frac{\mu_1+\lambda\alpha\sigma_1^2}{\sigma_1}\Big)^{-2}\frac{1}{s}
+ O\!\big(s^{-3/2}\big)\Bigg]\,.
\end{align}
Finally, taking the $1/\alpha$ power and simplifying yields the expression 
\begin{align}
\left(\mathbb{E}_p\left[\left(\frac{c(\utheta)}{f(\utheta)}\right)^\alpha\right]\right)^{1/\alpha} = e^{s\left(\lambda\mu_1+ \alpha\frac{\lambda^2 \sigma_1^2}{2}\right)}
s\left(\mu_1+ \lambda \alpha \sigma_1^2\right)  \Big[1+ \frac{\alpha(\alpha-1)}{2} \Big(\frac{\mu_1+ \lambda \alpha \sigma_1^2}{\sigma_1}\Big)^{-2}\frac{1}{s} + O\!\big(s^{-3/2}\big)\Big]^{1/\alpha}.\label{eq:netcost_gen_alpha}
\end{align}
First, we notice that, as long as $\lambda \neq 0$
and either $\mu_1 \neq 0$ or $\sigma_1 \neq 0$,
the expression showcases an exponential behavior as a function of $s$.
Moreover, in the $L=\infty$ case, we can readily use it to compute the ratio
\begin{align}
    \frac{NC_{q^*}}{NC_p}(s) &= e^{-s\left(\lambda\sigma_1\over 2\right)^2}\,\frac{\mu_1+ \frac{\lambda  \sigma_1^2}{2}}{\mu_1+ \lambda \sigma_1^2}\,\left[1-\frac{1}{4s}\left(\frac{\sigma_1}{\mu_1+ \frac{\lambda \sigma_1^2}{2}}\right)^{2}+ O\!\big(s^{-3/2}\big) \right]\,.\;
\end{align}
In particular, if $\lambda = 0$ then $NC_{q^*}/NC_p \to 1$,
while for all $\lambda > 0$ one obtains $NC_{q^*}/NC_p \to 0$.
Observe that $NC_{q^*}/NC_p$ always has a single minimum,
precisely at
\begin{align}
    s^*(L = \infty) &= \frac{1}{8}\left(\frac{\sigma_1^2}{\mu_1 + \frac{\lambda \sigma_1^2}{2}} \right)^2  \left[1+ \sqrt{1+\frac{64}{\lambda^2\sigma_1^4}\left( \mu_1 + \frac{\lambda \sigma_1^2}{2}\right)^2} \right]\,.
\end{align}

\section{Proof of Theorem \ref{thm:Optimal_Band_limited_Randomized_Sampling}}\label{sec:proof_Optimal_Band_limited_Randomized_Sampling}
Our starting point is to impose constraints on the distribution $p$, 
which remains otherwise arbitrary but must satisfy the conditions discussed in Section~\ref{sec : Eigenpath Transversal}. 
The first condition is non-negativity. To enforce this, we write 
\begin{equation} p(t) = u^*(t) u(t), \end{equation} 
thus transforming the problem of determining $p$ into the problem of determining a complex-valued function $u(t)$. The second condition is normalization, which requires 
\begin{equation} 
\int dt\, p(t) = \int dt\, |u(t)|^2 = 1. \end{equation} 
Together, these two conditions guarantee that $p(t)$ defines a valid probability density function. 
Finally, the third condition arises from the requirement that the Fourier transform $\hat{p}(\omega)$ vanishes for all $|\omega| \geq \Delta$. Expressed in terms of $u$, this reads 
\begin{equation} 
\hat{p}(\omega) = \frac{1}{2 \pi}\int d\omega_1 \, \hat{u}^*(\omega_1) \hat{u}(\omega+\omega_1), 
\end{equation} so it suffices to impose 
\begin{equation} 
\hat{u}(\omega) = 0 \qquad \forall\, |\omega| \geq \Delta/2 
\end{equation} 
to satisfy the third constraint.

We now proceed to the next stage, 
namely expressing the net cost in terms of $u(t)$.
For a generic function $z$ of $\abs{t}$, we aim to minimize $\mathbb{E}_p[z(|t|)]$ .
Because of the normalization and band-limiting conditions, the Parseval theorem ensures that $\hat{u}(\omega)$ can be expanded on $[-\Delta/2, \Delta/2]$ as
\begin{align*}
    \hat{u}(\omega) &= \sum_{\text{odd } m} \gamma_m \cos\left(\frac{\pi}{\Delta} m \omega\right) + \sum_{\text{even } m} \gamma_m \sin\left(\frac{\pi}{\Delta} m \omega\right) \\
    &= \sum_{n=1}^\infty \left[ \gamma_{2n-1} \cos\left(\frac{2\pi}{\Delta} \left(n-\frac{1}{2}\right)\omega\right) + \gamma_{2n} \sin\left(\frac{2\pi}{\Delta} n\,\omega\right) \right],
\end{align*}
for suitable coefficients $\{\gamma_m\}$, whose dimension is $\lfloor \Delta^{1/2} \rfloor$. Taking the inverse Fourier transform gives
\begin{equation*}
    u(t) = \frac{1}{2\pi} \int_{-\Delta/2}^{\Delta/2} d\omega\, \hat{u}(\omega) e^{i\omega t}.
\end{equation*}
Using standard integral identities, this can be rewritten in the compact form
\begin{align*}
    u(t) 
    &= \frac{1}{\pi} \sum_{n=1}^\infty \Bigg[
        \frac{\gamma_{2n-1}}{t^2-(2\pi/\Delta)^2(n-\frac{1}{2})^2} \frac{2\pi}{\Delta} \left(n-\frac{1}{2}\right) \cos\left(\frac{\Delta t}{2}\right) \sin\left(-\pi \left(n-\frac{1}{2}+n\right)\right) \\
    &\quad + \frac{\gamma_{2n}}{t^2-(2\pi/\Delta)^2n^2} \frac{2\pi}{\Delta} n\, \sin\left(\frac{\Delta t}{2}\right) \cos\left(\pi n\right)
    \Bigg].
\end{align*}
Introducing suitable adimensional coefficients $a_n$ and $b_n$, the general form becomes
\begin{equation}
    u(t) = 
    \left[\sum_{n=1}^\infty \frac{4\sqrt{\pi\Delta}\left(n-\frac{1}{2}\right)\, a_n}{\Delta^2t^2-(2\pi)^2(n-\tfrac{1}{2})^2}\right] \cos\left(\frac{\Delta t}{2}\right)
    +
    \left[\sum_{n=1}^\infty \frac{4\sqrt{\pi\Delta}\,n\, b_n}{\Delta^2t^2-(2\pi)^2\,n^2}\right] \sin\left(\frac{\Delta t}{2}\right)\,,
\end{equation}
from which, direct computation gives
\begin{align*}
    p(t) = |u(t)|^2
    &= \sum_{n,m=1}^\infty \frac{16\pi\Delta\left(n-\frac{1}{2}\right)\left(m-\frac{1}{2}\right) a_n^* a_m }{(\Delta^2t^2-(2\pi)^2(n-\tfrac{1}{2})^2)(\Delta^2t^2-(2\pi)^2(m-\tfrac{1}{2})^2)}\cos^2(\Delta t/2) \\
    &\quad + \sum_{n,m=1}^\infty \frac{16\pi\Delta\,n\,m\, b_n^* b_m }{(\Delta^2t^2-(2\pi)^2\,n^2)(\Delta^2t^2-(2\pi)^2\,m^2)}\sin^2(\Delta t/2) \\
    &\quad + \sum_{n,m=1}^\infty \frac{16\pi\Delta\left(n-\frac{1}{2}\right)\,m \Re(a_n^* b_m) }{(\Delta^2t^2-(2\pi)^2\left(n-\frac{1}{2}\right)^2)(\Delta^2t^2-(2\pi)^2\,m^2)}\sin(\Delta t).
\end{align*}
Since $z(|t|)$ is an even function of $t$, the last term does not contribute to
\begin{align}
    \mathbb{E}_p[z(|t|)] = \int dt\, p(t) z(|t|)\,.
\end{align}
Defining
\begin{align}
    \mathds{A}_{nm}(z,\Delta) &= 
    16\pi \Delta \Big(n-\tfrac{1}{2}\Big)\Big(m-\tfrac{1}{2}\Big)
    \int dt\,
    \frac{z(|t|) \cos^2(\Delta t/2)}
    {\bigl(\Delta^2t^2-(2\pi)^2(n-\tfrac{1}{2})^2\bigr)
     \bigl(\Delta^2t^2-(2\pi)^2(m-\tfrac{1}{2})^2\bigr)}, \\[2mm]
    \mathds{B}_{nm}(z,\Delta) &= 
    16\pi \Delta n\,m
    \int dt\,
    \frac{z(|t|) \sin^2(\Delta t/2)}
    {\bigl(\Delta^2t^2-(2\pi)^2n^2\bigr)
     \bigl(\Delta^2t^2-(2\pi)^2m^2\bigr)}\,,
\end{align}
then, in matrix notation,
we can write
\begin{align}
    \mathbb{E}_p[z(|t|)] =\underline{a}^\dagger \, \mathds{A}(z,\Delta) \, \underline{a} + \underline{b}^\dagger \, \mathds{B}(z,\Delta) \, \underline{b}\,.
\end{align}
More compactly,
we can collect all the elements of $\underline{a}$ and $\underline{b}$ in a new vector
\begin{equation}
    \underline{v} := \left[ \underline{a},\, \underline{b}\right]
\end{equation}
and define
\begin{equation}
    \mathds{V}(z,\Delta) := \begin{bmatrix}
        \mathds{A}(z,\Delta) & \mathds{O}\\
        \mathds{O} & \mathds{B}(z,\Delta)
    \end{bmatrix}
\end{equation}
to write
\begin{align}
    \mathbb{E}_p[z(|t|)] =\underline{v}^\dagger \, \mathds{V}(z,\Delta) \, \underline{v} \,.
\end{align}
so that
the optimization problem reduces to the minimization of 
\begin{align}
    \mathbb{E}_p[z(|t|)] = \underline{v}^\dagger \, \mathds{V}(\alpha k,\Delta) \, \underline{v}\,,
\end{align}
under the normalization constraint
\begin{align}
    \mathbb{E}_p[1] = \underline{v}^\dagger \, \mathds{V}(1,\Delta) \, \underline{v} = 1\,.
\end{align}
It is possible to prove that (see Appendix~\ref{appx:A_B_computation})
\begin{align}
    \mathds{V}(1,\Delta) = \mathds{1}\,,
\end{align}
hence the constraint translates in the normalization condition
\begin{align}
    |\underline{v}|^2  = 1.
\end{align}
Therefore the minimization problem reduces to an eigenvalue problem: 
the optimal vector 
\begin{align}
    \underline{\tilde{v}} = \left[\underline{\tilde{a}},\, \underline{\tilde{b}}\right]
\end{align}
is the normalized eigenvector of $\mathds{V}(z,\Delta)$ 
corresponding to its smallest eigenvalue. 
Since $\mathds{V}(z,\Delta)$ is block-diagonal,
the solution is either of the form
\begin{align}
    \underline{\tilde{v}} = \left[\underline{\tilde{a}},\, \underline{0}\right]
\end{align}
or
\begin{align}
    \underline{\tilde{v}} = \left[\underline{0}, \underline{\tilde{b}}\right]\,.
\end{align}
The final optimal distribution is then
\begin{align}
    p(t) &= 
    \sum_{n,m=1}^\infty \frac{16\pi\Delta\left(n-\frac{1}{2}\right)\left(m-\frac{1}{2}\right) \tilde{a}_n^* \tilde{a}_m }{(\Delta^2t^2-(2\pi)^2(n-\tfrac{1}{2})^2)(\Delta^2t^2-(2\pi)^2(m-\tfrac{1}{2})^2)}\cos^2(\Delta t/2) \\
    &\quad + \sum_{n,m=1}^\infty \frac{16\pi\Delta\,n\,m\, \tilde{b}_n^* \tilde{b}_m }{(\Delta^2t^2-(2\pi)^2\,n^2)(\Delta^2t^2-(2\pi)^2\,m^2)}\sin^2(\Delta t/2)
\end{align}

\subsection{Computation of $\mathds{A}$ and $\mathds{B}$}
\label{appx:A_B_computation}
In this section we report the analytical expressions for the integrals $\mathds{A}_{nm}(z,\Delta)$ and $\mathds{B}_{nm}(z,\Delta)$ in Eq.~\eqref{eq:A_B_definition},
for 
\begin{align}
    z(\abs{t}) = 1\,, \qquad z(\abs{t}) = \abs{t}\,, \qquad z(\abs{t}) = t^2\,.
\end{align}
We start proving that
\begin{align}
    \boxed{\mathds{A}_{nm}(1, \Delta) = \mathds{B}_{nm}(1, \Delta) = \delta_{nm}}
\end{align}
\begin{proof}
    Defining 
        \begin{align}
            \mu_n = 2\pi(n-\tfrac{1}{2})\,, \hspace{2cm} \nu_n = 2\pi n\,,
        \end{align}
    and using the change of variable $t \mapsto \tau = \Delta t$ the expressions become
    \begin{align}
    \mathds{A}_{nm}(1,\Delta) &= 
    16\pi \Big(n-\tfrac{1}{2}\Big)\Big(m-\tfrac{1}{2}\Big)
    \int d\tau\,
    \frac{\cos^2(\tau/2)}
    {\bigl(\tau^2-\mu_n^2\bigr)
     \bigl(\tau^2-\mu_m^2\bigr)}\,,\\
    \mathds{B}_{nm}(1,\Delta) &= 
    16\pi  \,n\,m\,
    \int dt\,
    \frac{ \sin^2(\tau/2)}
    {\bigl(\tau^2-\nu_n^2\bigr)
     \bigl(\tau^2-\nu_m^2\bigr)}\,.
    \end{align}
    We first consider
    \begin{align}
        \mathds{A}_{nm}(1,\Delta) &= 16\pi \Big(n-\tfrac{1}{2}\Big)\Big(m-\tfrac{1}{2}\Big)
        \int d\tau\,
        \frac{ \cos^2(\tau/2)}
        {\bigl(\tau^2-\mu_n^2\bigr)
         \bigl(\tau^2-\mu_m^2\bigr)}\, .
    \end{align}
    and rewrite the integral as
    \begin{align}
        \int d\tau'\,
        \frac{ \cos^2(\tau/2)}
        {\bigl(\tau^2-\mu_n^2\bigr)
         \bigl(\tau^2-\mu_m^2\bigr)}
         &= \frac{1}{ \mu_m^2-\mu_n^2}\int d\tau'\,
        \cos^2(\tau/2)\left[\frac{1}{\tau^2-\mu_m^2}- \frac{1}{\tau^2-\mu_n^2}\right]\\
        &= \frac{I(\mu_m)- I(\mu_n)}{ \mu_m^2-\mu_n^2} \,,
    \end{align}
    Where we defined
    \begin{align}
        I(\mu) = \int d\tau'\,
        \frac{\cos^2(\tau/2)}{\tau^2-\mu^2}\,.
    \end{align}
    The integral $I(\mu)$is understood in the Cauchy principal value sense over $\mathbb{R}$.
    Using $\cos^2(\tau/2)=\tfrac12(1+\cos\tau)$, we write
    \begin{align}
    \mathrm{P.V.}\,I(\mu)
    = \frac12 \mathrm{P.V.}\int_{-\infty}^{\infty}
    \left[\frac{1}{\tau^2-\mu^2}+\frac{\cos\tau}{\tau^2-\mu^2}\right]d\tau .
    \end{align}
    The first term vanishes since
    \begin{align}
    \frac{1}{\tau^2-\mu^2}
    = \frac{1}{2\mu}\left(\frac{1}{\tau-\mu}-\frac{1}{\tau+\mu}\right),
    \end{align}
    \begin{align}
    \mathrm{P.V.}\int_{-\infty}^{\infty}\frac{d\tau}{\tau^2-\mu^2}=0 .
    \end{align}
    For the second term, consider
    \begin{align}
    \mathrm{P.V.}\int_{-\infty}^{\infty}\frac{e^{i\tau}}{\tau^2-\mu^2}\,d\tau
    \end{align}
    and close the contour in the upper half of the complex plane.
    The simple poles at $\tau=\pm\mu$ lie on the real axis and each contributes half its residue, yielding
    \begin{align}
    \mathrm{P.V.}\int_{-\infty}^{\infty}\frac{e^{i\tau}}{\tau^2-\mu^2}\,d\tau
    = -\pi\frac{\sin\mu}{\mu}.
    \end{align}
    Taking the real part gives
    \begin{align}
    \mathrm{P.V.}\int_{-\infty}^{\infty}\frac{\cos\tau}{\tau^2-\mu^2}\,d\tau
    = -\pi\frac{\sin\mu}{\mu}.
    \end{align}
    Therefore
    \begin{align}
    \mathrm{P.V.}\,I(\mu)= -\frac{\pi}{2}\frac{\sin\mu}{\mu}
    \end{align}
    and
    \begin{align}
        \mathds{A}_{nm}(1,\Delta) &= -16\pi \Big(n-\tfrac{1}{2}\Big)\Big(m-\tfrac{1}{2}\Big)\frac{\pi}{2}\frac{\frac{\sin(\mu_m)}{\mu_m}- \frac{\sin(\mu_n)}{\mu_n}}{\mu_m^2-\mu_n^2}\,.
    \end{align}
    For $m \neq n$ one has 
    \begin{equation}
        \sin(\mu_n) = \sin(\mu_m) = 0\,,
    \end{equation}
    and hence $\mathds{A}_{nm}(1,\Delta)$ vanishes.
    To study the case of $m = n$ we can take the limit of $\mu_m \to \mu_n$ and use the theorem of L'Hôpital to write
    \begin{align}
        \mathds{A}_{nn}(1,\Delta) &= -\lim_{\mu_m \to \mu_n}16\pi \Big(n-\tfrac{1}{2}\Big)^2\frac{\pi}{2}\frac{\frac{\sin(\mu_m)}{\mu_m}- \frac{\sin(\mu_n)}{\mu_n}}{(\mu_m-\mu_n)(\mu_m+\mu_n)}\\
        &= - 4\pi \Big(n-\tfrac{1}{2}\Big)^2\frac{\pi}{\mu_n}\frac{d}{d\mu}\left[\frac{\sin(\mu)}{\mu}\right]\Bigg|_{\mu_n}\\
        &= - 4\pi \Big(n-\tfrac{1}{2}\Big)^2\frac{\pi}{\mu_n}\frac{\cos(\mu_m)\mu_m- \sin(\mu_m)}{\mu_m^2}\\
        &= - 4\pi^2 \Big(n-\tfrac{1}{2}\Big)^2\frac{ \cos(\mu_m)}{\mu_n^2}\\
        &= 1\,.
    \end{align}
Similarly, we now focus on
\begin{align}
    \mathds{B}_{nm}(1,\Delta) &= 16\pi \,n\,m\,
    \int d\tau\,
    \frac{\sin^2(\tau/2)}
    {(\tau^2-\nu_n^2)(\tau^2-\nu_m^2)}\, .
\end{align}
Rewriting the integral using partial fractions gives
\begin{align}
\int d\tau\,
\frac{\sin^2(\tau/2)}
{(\tau^2-\nu_n^2)(\tau^2-\nu_m^2)}
= \frac{J(\nu_m)-J(\nu_n)}{\nu_m^2-\nu_n^2},
\end{align}
where we define
\begin{align}
J(\nu) = \int d\tau\, \frac{\sin^2(\tau/2)}{\tau^2-\nu^2}.
\end{align}
Using the identity $\sin^2(\tau/2) = \frac12 (1-\cos\tau)$, the integral becomes
\begin{align}
\mathrm{P.V.}\,J(\nu)
= \frac12 \mathrm{P.V.}\int_{-\infty}^{\infty} \frac{1 - \cos\tau}{\tau^2 - \nu^2}\, d\tau
= \frac12 \mathrm{P.V.}\int_{-\infty}^{\infty} \frac{d\tau}{\tau^2 - \nu^2}
- \frac12 \mathrm{P.V.}\int_{-\infty}^{\infty} \frac{\cos\tau}{\tau^2 - \nu^2}\, d\tau.
\end{align}
The reminder of the proof is equivalent to the one for $\mathds{A}(1, \Delta)$
\end{proof}

We now consider the case $z(|t|) = |t|$, 
for which we report the analytical results obtained using symbolic calculus,
namely
\begin{equation}
     \boxed{\mathds{A}_{nm}(|t|, \Delta)
=
\begin{cases}
\dfrac{(2n+1)\pi\,\mathrm{Si}\!\bigl(\pi(2n+1)\bigr)-2}
      {2\pi^2(2n+1)^2},
& n=m, \\[1.2ex]
\displaystyle
-\dfrac{
\mathrm{Ci}\!\bigl(\pi(2n+1)\bigr)
-\mathrm{Ci}\!\bigl(\pi(2m+1)\bigr)
+\log\!\bigl(\tfrac{2m+1}{2n+1}\bigr)}
{\pi^2\!\bigl((2m+1)^2-(2n+1)^2\bigr)},
& n\neq m .
\end{cases}}
\end{equation}
and
\begin{align}
\boxed{\mathds{B}_{nm}(|t|, \Delta)
=
\begin{cases}
\dfrac{2\pi(n+1)\,\mathrm{Si}\!\bigl(2\pi(n+1)\bigr)-2}
      {8\pi^2(n+1)^2},
& n=m, \\[1.2ex]
\displaystyle
-\dfrac{
\mathrm{Ci}\!\bigl(2\pi(n+1)\bigr)
-\mathrm{Ci}\!\bigl(2\pi(m+1)\bigr)
+\log\!\bigl(\tfrac{m+1}{n+1}\bigr)}
{4\pi^2\!\bigl((m+1)^2-(n+1)^2\bigr)},
& n\neq m .
\end{cases}}
\end{align}
Finally, we prove that 
\begin{align}
    \boxed{\mathds{A}_{nm}(t^2, \Delta) = \left(\frac{2\pi}{\Delta}\right)^2 \Big(n-\tfrac{1}{2}\Big)^2 \delta_{nm}\,,}
\end{align}
and 
\begin{align}
    \boxed{\mathds{B}_{nm}(t^2, \Delta) = \left(\frac{2\pi}{\Delta}\right)^2 n^2 \,\delta_{nm}\,.}
\end{align}

\begin{proof}
Defining
\begin{align}
    \mu_n = 2\pi(n-\tfrac{1}{2})\,, \hspace{2cm} \nu_n = 2\pi n\,,
\end{align}
and using the change of variable $t \mapsto \tau = \Delta t$ the expressions become
\begin{align}
\mathds{A}_{nm}(t^2,\Delta) &= 
\frac{16\pi}{\Delta^2} \Big(n-\tfrac{1}{2}\Big)\Big(m-\tfrac{1}{2}\Big)
\int d\tau\,
\frac{\tau^2 \cos^2(\tau/2)}
{\bigl(\tau^2-\mu_n^2\bigr)
 \bigl(\tau^2-\mu_m^2\bigr)}\,,\\
\mathds{B}_{nm}(t^2,\Delta) &= 
\frac{16\pi}{\Delta^2} \,n\,m\,
\int d\tau\,
\frac{\tau^2 \sin^2(\tau/2)}
{\bigl(\tau^2-\nu_n^2\bigr)
 \bigl(\tau^2-\nu_m^2\bigr)}\,.
\end{align}
We first consider
\begin{align}
    \mathds{A}_{nm}(t^2,\Delta) &= \frac{16\pi}{\Delta^2} \Big(n-\tfrac{1}{2}\Big)\Big(m-\tfrac{1}{2}\Big)
    \int d\tau\,
    \frac{\tau^2 \cos^2(\tau/2)}
    {\bigl(\tau^2-\mu_n^2\bigr)
     \bigl(\tau^2-\mu_m^2\bigr)}\,
\end{align}
and rewrite the integral as
\begin{align}
\int d\tau\,
\frac{\tau^2 \cos^2(\tau/2)}
{(\tau^2-\mu_n^2)(\tau^2-\mu_m^2)}
= \frac{\mu_m^2 I(\mu_m)-\mu_n^2 I(\mu_n)}{\mu_m^2-\mu_n^2}\,,
\end{align}
where we defined
\begin{align}
I(\mu) = \mathrm{P.V.}\int_{-\infty}^{\infty} d\tau\,
\frac{\cos^2(\tau/2)}{\tau^2-\mu^2}\,.
\end{align}
Using $\cos^2(\tau/2)=\tfrac12(1+\cos\tau)$, we write
\begin{align}
\mathrm{P.V.}\,I(\mu)
= \frac12 \mathrm{P.V.}\int_{-\infty}^{\infty}
\left[\frac{1}{\tau^2-\mu^2} + \frac{\cos\tau}{\tau^2-\mu^2}\right]d\tau .
\end{align}
The first term vanishes since
\begin{align}
\mathrm{P.V.}\int_{-\infty}^{\infty}\frac{d\tau}{\tau^2-\mu^2}=0\,,
\end{align}
and the second term evaluates as before using a contour in the upper half-plane:
\begin{align}
\mathrm{P.V.}\int_{-\infty}^{\infty}\frac{\cos\tau}{\tau^2-\mu^2}\,d\tau
= -\pi\frac{\sin\mu}{\mu}.
\end{align}
Therefore
\begin{align}
I(\mu) = -\frac{\pi}{2}\frac{\sin\mu}{\mu},
\end{align}
and
\begin{align}
\mathds{A}_{nm}(t^2,\Delta) &= -\frac{16\pi}{\Delta^2} \Big(n-\tfrac12\Big)\Big(m-\tfrac12\Big) 
\frac{\mu_m^2 \frac{\sin\mu_m}{\mu_m} - \mu_n^2 \frac{\sin\mu_n}{\mu_n}}{\mu_m^2 - \mu_n^2}\\
&= -\frac{8\pi^2}{\Delta^2} \Big(n-\tfrac12\Big)\Big(m-\tfrac12\Big) 
\frac{\mu_m \sin\mu_m - \mu_n \sin\mu_n}{\mu_m^2 - \mu_n^2}\,.
\end{align}
For $m \neq n$ one has 
\begin{equation}
\sin(\mu_n) = \sin(\mu_m) = 0\,,
\end{equation}
and hence $\mathds{A}_{nm}(t^2,\Delta)$ vanishes.
To study the case of $m = n$, we take the limit $\mu_m \to \mu_n$ and use L'Hôpital's rule:
\begin{align}
\mathds{A}_{nn}(t^2,\Delta) &= -\lim_{\mu_m \to \mu_n} 16\pi \frac{1}{\Delta^2} \Big(n-\tfrac12\Big)^2 
\frac{\mu_m \sin\mu_m - \mu_n \sin\mu_n}{(\mu_m-\mu_n)(\mu_m+\mu_n)}\\
&= - \frac{8\pi}{\Delta^2} \Big(n-\tfrac12\Big)^2 \frac{1}{\mu_n} \frac{d}{d\mu} (\mu \sin\mu) \Big|_{\mu_n} \\
&= - \frac{8\pi}{\Delta^2} \Big(n-\tfrac12\Big)^2 \frac{\cos\mu_n \mu_n + \sin\mu_n}{\mu_n} \\
&= \left(\frac{2\pi}{\Delta}\right)^2 \Big(n-\tfrac{1}{2}\Big)^2 .
\end{align}
Similarly, we henceforth focus on
\begin{align}
\mathds{B}_{nm}(t^2,\Delta) &= \frac{16\pi}{\Delta^2} \,n\,m\,
\int d\tau\,
\frac{\tau^2 \sin^2(\tau/2)}
{(\tau^2-\nu_n^2)(\tau^2-\nu_m^2)}\,.
\end{align}
Rewriting the integral using partial fractions gives
\begin{align}
\int d\tau\,
\frac{\tau^2 \sin^2(\tau/2)}
{(\tau^2-\nu_n^2)(\tau^2-\nu_m^2)}
= \frac{\nu_m^2 J(\nu_m) - \nu_n^2 J(\nu_n)}{\nu_m^2 - \nu_n^2},
\end{align}
where we define
\begin{align}
J(\nu) = \mathrm{P.V.}\int d\tau\, \frac{\sin^2(\tau/2)}{\tau^2-\nu^2}.
\end{align}
Using the identity $\sin^2(\tau/2) = \frac12 (1-\cos\tau)$, the integral becomes
\begin{align}
\mathrm{P.V.}\,J(\nu)
= \frac12 \mathrm{P.V.}\int_{-\infty}^{\infty} \frac{1 - \cos\tau}{\tau^2 - \nu^2}\, d\tau
= \frac{\pi}{2}\frac{\sin\nu}{\nu}.
\end{align}
Therefore, following the same steps as for $\mathds{A}$,
the proof is completed.
\end{proof}

\end{document}